\documentclass[11pt]{article}

\usepackage{mdframed}
\usepackage{comment}

\usepackage[utf8]{inputenc}
\usepackage{authblk}
\usepackage{color}

\def\showauthornotes{1}
\def\showtableofcontents{0}
\def\showkeys{0}
\def\showdraftbox{0}
\def\showcolorlinks{1}
\def\usemicrotype{0}
\def\showfixme{0}

\usepackage{etex}

\usepackage[utf8]{inputenc}

\usepackage{xspace,enumerate}

\usepackage[]{xcolor}

\usepackage[T1]{fontenc}
\usepackage[full]{textcomp}

\usepackage[american]{babel}

\usepackage{mathtools}

\usepackage{bm}

\usepackage{amsthm}

\newtheorem{theorem}{Theorem}[section]
\newtheorem*{theorem*}{Theorem}

\newtheorem{maintheorem}[theorem]{Main Theorem}

\newtheorem*{proposition*}{Proposition}
\newtheorem{lemma}[theorem]{Lemma}
\newtheorem*{lemma*}{Lemma}
\newtheorem{corollary}[theorem]{Corollary}
\newtheorem*{conjecture*}{Conjecture}

\newtheorem*{fact*}{Fact}

\newtheorem*{hypothesis*}{Hypothesis}
\newtheorem{conjecture}[theorem]{Conjecture}

\theoremstyle{definition}
\newtheorem{definition}[theorem]{Definition}

\newtheorem{problem}[theorem]{Problem}

\theoremstyle{remark}
\newtheorem{claim}[theorem]{Claim}
\newtheorem*{claim*}{Claim}
\newtheorem{remark}[theorem]{Remark}
\newtheorem*{remark*}{Remark}
\newtheorem{observation}[theorem]{Observation}
\newtheorem*{observation*}{Observation}

\usepackage[margin=1in]{geometry}

\usepackage[varg]{pxfonts} % varg - uses nicer g,v,y,w

\ifnum\showkeys=1
\usepackage[color]{showkeys}
\fi

\definecolor{OliveGreen}{rgb}{0,0.6,0}
\ifnum\showcolorlinks=1
\usepackage[
pagebackref,
colorlinks=true,
urlcolor=blue,
linkcolor=blue,
citecolor=OliveGreen,
]{hyperref}
\fi

\ifnum\showcolorlinks=0
\usepackage[
pagebackref,
colorlinks=false,
pdfborder={0 0 0}
]{hyperref}
\fi

\usepackage{prettyref}

\newcommand{\savehyperref}[2]{\texorpdfstring{\hyperref[#1]{#2}}{#2}}

\newrefformat{eq}{\savehyperref{#1}{\textup{(\ref*{#1})}}}
\newrefformat{lem}{\savehyperref{#1}{Lemma~\ref*{#1}}}
\newrefformat{def}{\savehyperref{#1}{Definition~\ref*{#1}}}
\newrefformat{thm}{\savehyperref{#1}{Theorem~\ref*{#1}}}
\newrefformat{cor}{\savehyperref{#1}{Corollary~\ref*{#1}}}
\newrefformat{cha}{\savehyperref{#1}{Chapter~\ref*{#1}}}
\newrefformat{sec}{\savehyperref{#1}{Section~\ref*{#1}}}
\newrefformat{app}{\savehyperref{#1}{Appendix~\ref*{#1}}}
\newrefformat{tab}{\savehyperref{#1}{Table~\ref*{#1}}}
\newrefformat{fig}{\savehyperref{#1}{Figure~\ref*{#1}}}
\newrefformat{hyp}{\savehyperref{#1}{Hypothesis~\ref*{#1}}}
\newrefformat{alg}{\savehyperref{#1}{Algorithm~\ref*{#1}}}
\newrefformat{rem}{\savehyperref{#1}{Remark~\ref*{#1}}}
\newrefformat{item}{\savehyperref{#1}{Item~\ref*{#1}}}
\newrefformat{step}{\savehyperref{#1}{step~\ref*{#1}}}
\newrefformat{conj}{\savehyperref{#1}{Conjecture~\ref*{#1}}}
\newrefformat{fact}{\savehyperref{#1}{Fact~\ref*{#1}}}
\newrefformat{prop}{\savehyperref{#1}{Proposition~\ref*{#1}}}
\newrefformat{prob}{\savehyperref{#1}{Problem~\ref*{#1}}}
\newrefformat{claim}{\savehyperref{#1}{Claim~\ref*{#1}}}
\newrefformat{relax}{\savehyperref{#1}{Relaxation~\ref*{#1}}}
\newrefformat{red}{\savehyperref{#1}{Reduction~\ref*{#1}}}
\newrefformat{part}{\savehyperref{#1}{Part~\ref*{#1}}}

\newcommand{\Sref}[1]{\hyperref[#1]{\S\ref*{#1}}}

\usepackage{nicefrac}
\ifnum\usemicrotype=1
\usepackage{microtype}
\fi

\ifnum\showauthornotes=1
\newcommand{\Authornote}[2]{{\sffamily\small\color{red}{[#1: #2]}}}
\newcommand{\Authornotecolored}[3]{{\sffamily\small\color{#1}{[#2: #3]}}}
\newcommand{\Authorcomment}[2]{{\sffamily\small\color{gray}{[#1: #2]}}}
\newcommand{\Authorstartcomment}[1]{\sffamily\small\color{gray}[#1: }

\newcommand{\Authorfnote}[2]{\footnote{\color{red}{#1: #2}}}
\newcommand{\Authorfixme}[1]{\Authornote{#1}{\textbf{??}}}
\newcommand{\Authormarginmark}[1]{\marginpar{\textcolor{red}{\fbox{\Large #1:!}}}}
\else
\newcommand{\Authornote}[2]{}
\newcommand{\Authornotecolored}[3]{}
\newcommand{\Authorcomment}[2]{}
\newcommand{\Authorstartcomment}[1]{}

\newcommand{\Authorfnote}[2]{}
\newcommand{\Authorfixme}[1]{}
\newcommand{\Authormarginmark}[1]{}
\fi

\ifnum\showfixme=0

\fi

\usepackage{boxedminipage}

\newcommand{\Esymb}{\mathbb{E}}
\newcommand{\Psymb}{\mathbb{P}}

\DeclareMathOperator*{\E}{\Esymb}

\DeclareMathOperator*{\ProbOp}{\Psymb}

\renewcommand{\Pr}{\ProbOp}

 \usepackage{dsfont}
\usepackage{mathrsfs}

\newcommand{\textparen}[1]{\text{(#1)}}

\ifx\because\undefined
\newcommand{\because}[1]{\textparen{because #1}}
\else
\renewcommand{\because}[1]{\textparen{because #1}}
\fi

\newcommand\bdot\bullet

\DeclareMathOperator{\LP}{LP}
\DeclareMathOperator{\OPT}{OPT}

\newcommand{\Z}{\mathbb Z}
\newcommand{\N}{\mathbb N}
\newcommand{\R}{\mathbb R}

\newcommand{\problemmacro}[1]{\texorpdfstring{\textsc{#1}}{#1}\xspace}

\newcommand{\uniquegames}{\problemmacro{unique games}}
\newcommand{\maxcut}{\problemmacro{max cut}}

\newcommand{\vertexcover}{\problemmacro{vertex cover}}

\newcommand{\maxthreesat}{\problemmacro{max $3$-sat}}

\newcommand{\knapsack}{\problemmacro{knapsack}}

\newcommand{\cC}{\mathcal C}
\newcommand{\cD}{\mathcal D}

\newcommand{\cI}{\mathcal I}

\newcommand{\cU}{\mathcal U}

\renewcommand{\leq}{\leqslant}

\renewcommand{\geq}{\geqslant}

\ifnum\showdraftbox=1
\newcommand{\draftbox}{\begin{center}
  \fbox{%
    \begin{minipage}{2in}%
      \begin{center}%
          \Large\textsc{Working Draft}\\%
        Please do not distribute%
      \end{center}%
    \end{minipage}%
  }%
\end{center}
\vspace{0.2cm}}
\else
\newcommand{\draftbox}{}
\fi

\let\epsilon=\varepsilon

\numberwithin{equation}{section}

\newcommand{\MYstore}[2]{%
  \global\expandafter \def \csname MYMEMORY #1 \endcsname{#2}%
}

\newcommand{\MYload}[1]{%
  \csname MYMEMORY #1 \endcsname%
}

\newcommand{\MYnewlabel}[1]{%
  \newcommand\MYcurrentlabel{#1}%
  \MYoldlabel{#1}%
}

\newcommand{\MYdummylabel}[1]{}

\newcommand{\torestate}[1]{%
  \let\MYoldlabel\label%
  \let\label\MYnewlabel%
  #1%
  \MYstore{\MYcurrentlabel}{#1}%
  \let\label\MYoldlabel%
}

\newcommand{\restatetheorem}[1]{%
  \let\MYoldlabel\label
  \let\label\MYdummylabel
  \begin{theorem*}[Restatement of \prettyref{#1}]
    \MYload{#1}
  \end{theorem*}
  \let\label\MYoldlabel
}

\newcommand{\restatelemma}[1]{%
  \let\MYoldlabel\label
  \let\label\MYdummylabel
  \begin{lemma*}[Restatement of \prettyref{#1}]
    \MYload{#1}
  \end{lemma*}
  \let\label\MYoldlabel
}

\newcommand{\restateprop}[1]{%
  \let\MYoldlabel\label
  \let\label\MYdummylabel
  \begin{proposition*}[Restatement of \prettyref{#1}]
    \MYload{#1}
  \end{proposition*}
  \let\label\MYoldlabel
}

\newcommand{\restatefact}[1]{%
  \let\MYoldlabel\label
  \let\label\MYdummylabel
  \begin{fact*}[Restatement of \prettyref{#1}]
    \MYload{#1}
  \end{fact*}
  \let\label\MYoldlabel
}

\newcommand{\restate}[1]{%
  \let\MYoldlabel\label
  \let\label\MYdummylabel
  \MYload{#1}
  \let\label\MYoldlabel
}

\let\origparagraph\paragraph
\renewcommand{\paragraph}[1]{\origparagraph{#1.}}

\allowdisplaybreaks

\newcommand{\ekvertexcover}{\problemmacro{$q$-Uniform-Vertex-Cover}}
\newcommand{\independentset}{\problemmacro{independent set}}
\newcommand{\constraintsatisfaction}{\problemmacro{constraint satisfaction problems}}
\newcommand{\oFk}{\problemmacro{1F-CSP}}
\newcommand{\mc}[1]{\mathcal{#1}}
\newcommand{\fc}{\textbf{\textsf{fc}}}

\newenvironment{proofof}[1]{\emph{Proof of #1. }}{\hfill$\square$}
\newcommand{\Infl}[3]{\text{\normalfont Inf}_{#1}^{#2}({#3})}
\newcommand{\mE}[2]{\mathop{\mathbb{E}}_{#1} \left[ #2\right]}

\newcommand{\hide}[1]{}

\newcommand{\mf}{\mathfrak}
\newcommand{\Val}{\textrm{Val}}
\newcommand{\Cost}{\textrm{Cost}}

\title{\bfseries No Small Linear Program Approximates Vertex Cover 
within a Factor $2 - \epsilon$}

\author[1]{Abbas Bazzi\thanks{email: abbas.bazzi@epfl.ch}}
\author[2]{Samuel Fiorini\thanks{email: sfiorini@ulb.ac.be}}
\author[3]{Sebastian Pokutta\thanks{email: sebastian.pokutta@isye.gatech.edu}}
\author[1]{Ola Svensson\thanks{email: ola.svensson@epfl.ch}}

\affil[1]{EPFL, School of Computer and Communication Sciences}
\affil[2]{Universit\'e libre de Bruxelles, D\'epartement de Math\'ematique}
\affil[3]{Georgia Tech, ISyE}

\begin{document}

\begin{titlepage}
\maketitle

\begin{abstract}
The vertex cover problem is one of the most important and intensively 
studied combinatorial optimization problems. Khot and Regev~\cite{KR03,KR08} 
proved that the problem is NP-hard to approximate within a factor
$2 - \epsilon$, assuming the Unique Games Conjecture (UGC). This is
tight because the problem has an easy $2$-approximation algorithm.
Without resorting to the UGC, the best inapproximability result for the
problem is due to Dinur and Safra~\cite{DS02,DS05}: vertex cover is 
NP-hard to approximate within a factor $1.3606$. 

We prove the following unconditional result about linear programming (LP)
relaxations of the problem: every LP relaxation that approximates vertex 
cover within a factor $2-\epsilon$ has super-polynomially many 
inequalities. As a direct consequence of our methods, we also establish 
that LP relaxations (as well as SDP relaxations) that approximate the independent set problem within 
any constant factor have super-polynomial size.
\end{abstract}

{\small \textbf{Keywords:}
Extended formulations, Hardness of approximation, Independent set, Linear programming, Vertex cover.
}
\end{titlepage}

\draftbox

\clearpage

\ifnum\showtableofcontents=1
{
\tableofcontents
\thispagestyle{empty}
 }
\fi

\clearpage

%%%%%%%%%%%%%%%%%% Intro %%%%%%%%%%%%%%%%%%%%
\section{Introduction}

In this paper we prove tight  inapproximability results for
\vertexcover  with respect to linear programming relaxations of polynomial
size. \vertexcover is the following classic problem:
given a graph $G = (V,E)$ together with vertex costs $c_v \geq 0$, $v \in V$,
find a minimum cost set of vertices $U \subseteq V$ such that every edge has at
least one endpoint in $U$. Such a set of vertices meeting every edge is called
a \emph{vertex cover}. 

It is well known that the LP relaxation
\begin{equation}
\label{eq:VC-LP-basic}
\begin{array}{rll}
\min &\displaystyle \sum_{v \in V} c_v x_v\\
\text{s.t.}&x_u + x_v \geq 1 &\forall uv \in E\\
&0 \leq x_v \leq 1 &\forall v \in V
\end{array}
\end{equation}
approximates \vertexcover within a factor~$2$. (See e.g., 
Hochbaum~\cite{Hochbaum97} and the references therein.) This means 
that for every cost vector there exists a vertex cover whose cost 
is at most $2$ times the optimum value of the LP. In fact, the 
(global) \emph{integrality gap} of this LP relaxation, the worst-case 
ratio over all graphs and all cost vectors between the minimum cost of 
an integer solution and the minimum cost of a fractional solution, 
equals $2$.

One way to make the LP relaxation \eqref{eq:VC-LP-basic} stronger is
by adding valid inequalities. Here, a \emph{valid inequality} is a
linear inequality $\sum_{v \in V} a_v x_v \geq \beta$ that is
satisfied by every integral solution. Adding all possible valid
inequalities to \eqref{eq:VC-LP-basic} would clearly decrease the
integrality gap all the way from $2$ to $1$, and thus provide a
perfect LP formulation. However, this would also yield an LP that we
would not be able to write down or solve efficiently.  Hence, it is
necessary to restrict to more tangible families of valid inequalities.

For instance, if $C \subseteq V$ is the vertex set of an odd cycle 
in $G$, then $\sum_{v \in C} x_v \geqslant \frac{|C| + 1}{2}$ is 
a valid inequality for vertex covers, known as an \emph{odd cycle
inequality}. However, the integrality gap remains $2$ after adding 
all such inequalities to \eqref{eq:VC-LP-basic}. 
More classes of inequalities are known beyond the odd cycle 
inequalities. However, we do not know any such class of valid 
inequalities that would decrease the integrality gap strictly below $2$. 

There has also been a lot of success ruling out concrete
polynomial-size linear programming formulations arising from, e.g., the
addition of a polynomial number of inequalities with sparse support or
those arising from hierarchies, where new valid inequalities are generated
in a systematic way. For instance, what about adding all valid inequalities supported 
on at most $o(n)$ vertices (where $n$ denotes the number of vertices 
of $G$), or all those obtained by performing a few rounds of the 
Lov\'asz-Schrijver (LS) lift-and-project procedure~\cite{LS91}? In 
their influential paper Arora, Bollob\'as 
and Lov\'asz~\cite{ABL02} (the journal version \cite{ABLT06} is joint work 
with Tourlakis) proved that none of these broad classes of valid 
inequalities are sufficient to decrease the integrality gap to
$2 - \epsilon$ for any $\epsilon > 0$.

The paper of Arora~\emph{et al.}\ was followed by many papers deriving stronger
and stronger tradeoffs between number of rounds and integrality gap for
\vertexcover and many other problems in various hierarchies, see the related
work section below. The focus of this paper is to prove lower bounds in a more
general model. Specifically, our goal is to understand the strength of
\emph{any} polynomial-size linear programming relaxation of \vertexcover
independently of any hierarchy and irrespective of any
complexity-theoretic assumption such as e.g., \(P \neq NP\).

We will rule out \emph{all possible} polynomial-size LP relaxations 
obtained from adding an \emph{arbitrary} set of valid inequalities of 
polynomial size. By ``all possible LP relaxations'', we mean
that the variables of the LP can be chosen arbitrarily. They do not have
to have to be the vertex-variables of \eqref{eq:VC-LP-basic}.

\subsection*{Contribution}

We consider the general model of LP relaxations as in~\cite{CLRS13}, see
also~\cite{BPZ}. Given an $n$-vertex graph $G=(V,E)$, a system of linear 
inequalities $Ax \geqslant b$ in $\R^d$, where $d \in \N$ is arbitrary, 
defines an \emph{LP relaxation} of $\vertexcover$ (on $G$) if the following 
conditions hold:
\begin{description}
  \item[Feasibility:] For every vertex cover $U\subseteq V$, we have a feasible vector $x^U \in \R^d$ satisfying $Ax^U \geqslant b$.
  \item[Linear objective:] For every vertex-costs $c \in \R^V_+$, we have an affine function (degree-$1$ polynomial) $f_c : \R^d \to \R$.
  \item[Consistency:] For all vertex covers $U \subseteq V$ and vertex-costs $c \in \R^V_+$, the condition $f_c(x^U) = \sum_{v \in U} c_v$ holds. 
\end{description}
For every vertex-costs $c \in \R^V_+$, the LP $\min \{f_c(x) \mid 
Ax \geqslant b\}$ provides a guess on the minimum cost of a vertex
cover. This guess is always a lower bound on the optimum.

We allow arbitrary computations for writing down the LP, and do not
bound the size of the coefficients. We only care about the following 
two parameters and their relationship: the \emph{size} of the LP 
relaxation, defined as the number of inequalities in $Ax \geqslant b$, 
and the (graph-specific) \emph{integrality gap} which is the 
worst-case ratio over all vertex-costs between the true optimum and 
the guess provided by the LP, for this particular graph $G$ and LP 
relaxation.

This framework subsumes the polyhedral-pair approach in extended 
formulations~\cite{BFPS12}; see also \cite{Pashkovich12}. We refer 
the interested reader to the surveys \cite{CCZ10,Kaibel11} for an 
introduction to extended formulations; see also Section~\ref{LPSec} 
for more details.

In this paper, we prove the following result about LP relaxations
of \vertexcover and, as a byproduct, \independentset.\footnote{Recall that an 
\emph{independent set} (stable set) in graph $G = (V,E)$ is a 
set of vertices $I \subseteq V$ such that no edge has both endpoints 
in $I$. \independentset is the corresponding maximization problem: 
given a graph together with a weight for each vertex, find a maximum 
weight independent set.}

\begin{theorem} \label{thm:main}
For infinitely many values of $n$, there exists an $n$-vertex graph $G$
such that: 
(i) Every size-$n^{o\left( \log n / \log \log n \right)}$ LP relaxation 
of \vertexcover on $G$ has integrality gap $2-o(1)$; 
(ii) Every size-$n^{o\left( \log n / \log \log n \right)}$ LP 
relaxation of \independentset on $G$ has integrality gap \(\omega(1)\).
\end{theorem}

This solves an open problem that was posed both by Singh~\cite{Singh10} 
and Chan, Lee, Raghavendra and Steurer~\cite{CLRS13}. 
In fact, Singh conjectured that every compact (that is, polynomial size),
\emph{symmetric} extended formulation for \vertexcover has integrality 
gap at least $2 - \epsilon$. We prove that his conjecture holds, even if 
asymmetric extended formulations are allowed.\footnote{Note that in some 
cases imposing symmetry is a severe restriction, see Kaibel, Pashkovich 
and Theis~\cite{KPT10}.}

Our result for the \independentset problem is even stronger than 
Theorem~\ref{thm:main}, as we are also able to rule out any polynomial 
size SDP with constant integrality gap for this problem. Furthermore, 
combining our proof strategy with more complex techniques 
we can prove a result similar to Theorem~\ref{thm:main} for \ekvertexcover
(that is, vertex cover in $q$-uniform \emph{hypergraphs}), for any fixed $q \geq 2$. 
For that problem, every size $n^{o\left( \log n / \log \log n \right)}$ LP 
relaxation has integrality gap $q-o(1)$. This generalizes our result on 
(graph) \vertexcover.

In the general model of LP relaxations outlined above, the LPs are
designed with the knowledge of the graph $G=(V,E)$; this is a 
\emph{non-uniform} model as the LP can depend on
the graph. It captures the natural LP relaxations for \vertexcover and
\independentset whose constraints depend on the graph structure. This
is in contrast to previous lower bound results (\cite{BFPS12,BM13,BP13})
on the LP formulation complexity of \independentset, which are of a 
\emph{uniform} nature: In those works, the formulation of the LP 
relaxation was agnostic to the input graph and only
allowed to depend on the number of vertices of the graph; see
\cite{BPZ} for a discussion of uniformity vs.\ non-uniformity. In general
non-uniform models are stronger (and so are lower bounds for
it) and interestingly, this allows for stronger LP relaxations for
\independentset than NP-hardness would predict. This phenomenon
is related to the approximability of problems with preprocessing. In
Section~\ref{sec:upperbounds}, we observe that a result of Feige and
Jozeph~\cite{FJ14} implies that there exists a size-$O(n)$ LP
formulation for approximating \independentset within a multiplicative
factor of
$O(\sqrt{n})$.

\subsection*{Related work}

Most of the work on extended formulations is ultimately rooted in
Yannakakis's famous paper~\cite{Yannakakis88,Yannakakis91} in which he
proved that every symmetric extended formulation of the matching
polytope and (hence) TSP polytope of the $n$-vertex complete graph has
size $2^{\Omega(n)}$. Yannakakis's work was motivated by approaches to
proving P $=$ NP by providing small (symmetric) LPs for the TSP, which
he ruled out. 

The paper of Arora \emph{et al.}~\cite{ABL02,ABLT06} revived
Yannakakis's ideas in the context of hardness of approximation and
provided lower bounds for \vertexcover in LS. It marked the starting
point for a whole series of papers on approximations via
hierarchies. Shortly after Arora \emph{et al.}~proved that performing
$O(\log n)$ rounds of LS does not decrease the integrality gap below
$2$, Schoenebeck, Trevisan and Tourlakis~\cite{STT07} proved that this
also holds for $o(n)$ rounds of LS. A similar result holds for the
stronger Sherali-Adams (SA) hierarchy~\cite{SA90}: Charikar, Makarychev and Makarychev~\cite{CMM}
showed that $\Omega(n^\delta)$ rounds of SA are necessary to decrease
the integrality gap beyond $2 - \epsilon$ for some
$\delta = \delta(\epsilon) > 0$.  

Beyond linear programming hierarchies, there are also semidefinite programming
(SDP) hierarchies, e.g., Lov\'asz-Schrijver (LS+)~\cite{LS91} and
Sum-of-Squares/Lasserre~\cite{Parillo00,Lasserre01a,Lasserre01b}.  Georgiou,
Magen, Pitassi and Tourlakis~\cite{GMPT07} proved that $O(\sqrt{\log n / \log
\log n})$ rounds of LS+ does not approximate \vertexcover within a factor
better than $2$. In this paper, we focus mostly on the LP case.

Other papers in the ``hierarchies'' line of work 
include~\cite{FM07b,GM08,Schoenebeck08,KS09,RS09,Tulsiani09,KMN11,BCVGZ,BGMT12}.  

Although hierarchies are a powerful tool, they have their limitations. For 
instance, $o(n)$ rounds of SA does not give an approximation of 
\knapsack with a factor better than $2$~\cite{KMN11}. However, for every
$\epsilon > 0$, there exists a size-$n^{1/\epsilon + O(1)}$ LP 
relaxation that approximates \knapsack within a factor of 
$1 + \epsilon$~\cite{Bienstock08}. 

Besides the study of hierarchy approaches, there was a distinct line
of work inspired directly by Yannakakis's paper that sought to study the 
power of general (linear) extended formulations, independently of any hierarchy, 
see e.g., \cite{Rothvoss11,FMPTW12,BFPS12,BGMT12,BM13,BP13,Rothvoss13}. 
Limitations of semidefinite extended formulations were also studied recently, 
see~\cite{BDP13,LeeRS14}.

The lines of work on hierarchies and (general) extended formulations
in the case of \constraintsatisfaction (CSPs) were merged in 
the work of Chan \emph{et al.}~\cite{CLRS13}. Their main result states that for Max-CSPs, SA is best 
possible among all LP relaxations in the sense that if there exists a 
size-$n^{r}$ LP relaxation approximating a given Max-CSP within factor 
$\alpha$ then performing $2r$ rounds of SA would also provide a 
factor-$\alpha$ approximation. They obtained several strong LP 
inapproximability results for Max-CSPs such as \maxcut and 
\maxthreesat. This result was recently strengthened in a breakthrough by Lee, Raghavendra, and
Steurer~\cite{LeeRS14}, who obtained analogous results showing (informally) that the Sum-of-Squares/Lasserre hierarchy is best possible among all SDP relaxations for Max-CSPs.

Braun, Pokutta and Zink~\cite{BPZ} developed a framework for proving
size lower bounds on LP relaxations via reductions. Using~\cite{CLRS13}
and FGLSS graphs~\cite{FGLSS96}, they obtained a $n^{\Omega\left(\log n / 
\log \log n\right)}$ size lower bound for approximating \vertexcover within 
a factor of $1.5-\epsilon$ and \independentset within a factor of $2-\epsilon$. Our paper 
improves these inapproximability factors to a tight $2 - \epsilon$ and any 
constant, respectively.

\subsection*{Outline}

The framework in Braun \emph{et al.}~\cite{BPZ} formalizes sufficient
properties of reductions for preserving inapproximability with respect to
extended formulations / LP relaxations; this reduction mechanism does not 
capture all known reductions due to certain linearity and independence 
requirements. Using this framework, they gave a reduction from \maxcut
to \vertexcover yielding the aforementioned result.

A natural approach for strengthening the hardness factor is to reduce from
\uniquegames instead of \maxcut (since \vertexcover is known to be \uniquegames-hard to approximate within a factor $2-\epsilon$). However, one obstacle is
that, in known  reductions from \uniquegames, the optimal value of the obtained
\vertexcover instance is not \emph{linearly} related to the value of the
\uniquegames instance. This makes these reductions unsuitable for the framework
in~\cite{BPZ} (see Definition~\ref{RedDef}).

We overcome this obstacle by designing a two-step reduction. In the first
step (Section~\ref{sec-1fcsp}), we interpret the ``one free bit'' PCP test 
of Bansal and Khot~\cite{BK} as a reduction from a \uniquegames instance 
to a ``one free bit'' CSP (\oFk). We then use the family of SA integrality 
gap instances for the \uniquegames problem constructed by Charikar 
\emph{et al.}~\cite{CMM}, to construct a similar family for this CSP. This,
together with the main result of Chan \emph{et al.}~\cite{CLRS13} applied to
this particular CSP, implies that no size-$n^{o\left( \log n / \log \log
n \right)}$ LP relaxation can provide a constant factor approximation for \oFk.
In the second step (Section~\ref{LPSec}), a reduction from \oFk to \vertexcover, in
the framework of Braun \emph{et al.}~\cite{BPZ}, then yields our main result.

Finally, following a slightly different and more challenging route we 
prove  tight hardness of approximation for LP relaxations of \ekvertexcover
for every $q \geqslant 2$. This is done in Section~\ref{sec:qvertexcover}.

%%%%%%%%%%%%%%%%% Preliminaries %%%%%%%%%%%%%%%%%
\section{Preliminaries}

We shall now present required tools and background. In Sections
\ref{sec:csp} and \ref{sec:SA} we define the class of
\constraintsatisfaction and the Sherali-Adams (SA) hierarchy,
respectively.

\subsection{\constraintsatisfaction}
\label{sec:csp}

The class of \constraintsatisfaction (CSPs) captures a large variety of
combinatorial problems, like \maxcut and \maxthreesat. In general, we
are given a collection of \emph{predicates} $\mc{P} = \{P_1, \dots, P_m\}$ 
(or \emph{constraints} $\mc{C} = \{C_1,\ldots,C_m\}$) where 
each $P_i$ is of the form $P_i: [R]^n \mapsto \{0,1\}$, where 
$[R] := \{1,\ldots,R\}$ is the \emph{domain} and $n$ is the number of 
variables. We will be mainly interested in the family of CSPs where each 
predicate $P$ is associated with a set of distinct indices 
$S_P = \{i_1,\dots,i_k\} \subset [n]$ and is of constant arity
$k$, i.e., $P: [R]^k \mapsto \{0,1\}$. In this terminology, for
$x\in [R]^n$ we set $P(x) \coloneqq P(x_{i_1},\ldots,x_{i_k})$. The goal in such
problems is to find an assignment for $x\in [R]^n$ in such a way as to
maximize the total fraction of satisfied predicates.

The \emph{value} of an assignment $x \in [R]^n$ for a CSP instance 
$\mc{I}$ is defined as
$$
\textrm{Val}_\mc{I}(x) := \frac{1}{m} \sum_{i=1}^m P_i(x) = \mE{P \in \mc{P}}{P(x)},
$$
and the optimal value of such instance $\mc{I}$, denoted by 
$\textrm{OPT}(\mc{I})$ is $$\textrm{OPT}(\mc{I}) = 
\max_{x\in [R]^n} \textrm{Val}_\mc{I}(x).$$

Often, we will consider \emph{binary} CSPs, that is, with domain size
$R = 2$. Given a binary predicate $P: \{0,1\}^k \mapsto \{0,1\}$, the 
\emph{free bit} complexity of $P$ is defined to be $\log_2(|\{z \in 
\{0,1\}^k: P(z) = 1\}|)$. For example the \maxcut predicate $x_i \oplus 
x_j$ has a free bit complexity of \emph{one}, since the only \emph{two} 
accepting configurations are $(x_i = 0, x_j = 1)$ and $(x_i = 1, x_j = 0)$.

For the LP-hardness of \vertexcover and \independentset (i.e., Sections~\ref{sec-1fcsp} and \ref{LPSec}), we will be interested in a one free bit binary CSP, that we 
refer to as $\oFk$, defined as follows: 
\begin{definition}[\oFk]
\label{def:1FCSP}
A \emph{\oFk instance of arity $k$} is a binary CSP over a set of variables $\{x_1,\dots,x_n\}$ and a set of constraints $\mc{C} = \{C_1,\dots, C_m\}$ such 
that each constraint $C\in \mc{C}$ is of arity $k$ and has only two accepting
configurations out of the $2^k$ possible ones.
\end{definition}

%%%%%%%%%%%%%%%%%%%%%%%%%%%%%%%%%

\subsection{Sherali-Adams Hierarchy}
\label{sec:SA}

We define the  canonical relaxation for \constraintsatisfaction as it is
obtained by $r$-rounds of the Sherali-Adams (SA) hierarchy. We follow the
notation as in e.g.,~\cite{GeorgiouMT09}. For completeness we also describe in Appendix~\ref{sec:genSA} why
this relaxation is equivalent to the one obtained by applying the original
definition of SA as a reformulation-linearization technique on a binary program.

Consider any CSP defined over $n$ variables $x_1, \ldots, x_n \in [R]$, with
a set of $m$ constraints $\mc{C} = \{C_1,\dots,C_m\}$ where the arity of each constraint
is at most $k$. Let $S_i = S_{C_i}$ denote the set of variables that $C_i$ depends on. 
The $r$-rounds SA relaxation of this CSP has a variable $X_{(S,
\alpha)}$ for each $S\subseteq [n], \alpha \in [R]^S$ with $|S| \leq r$. The intuition is that
$X_{(S, \alpha)}$ models the indicator variable whether the variables in $S$ are
assigned the values in $\alpha$.  The $r$-rounds SA relaxation with $r\geq k$  is now 
\begin{equation}
\label{eq:CSP_prog2}
\begin{array}{rll}
  \max &\displaystyle \frac{1}{m} \sum_{i=1}^m \sum_{\alpha \in [R]^{S_i}} C_i(\alpha) \cdot X_{(S_i, \alpha)}\\[3ex]
  \text{s.t.}&\displaystyle  \sum_{u\in [R]} X_{(S \cup \{j\}, \alpha \circ u)}  = X_{(S,\alpha)}  \quad &\forall S \subseteq [n]: |S| < r, \alpha \in [R]^S, j \in [n]\setminus S\, ,\\
  &\qquad \qquad X_{(S, \alpha)} \geq 0 \qquad & \forall S \subseteq [n]: |S| \leq r, \alpha \in [R]^S\,, \\[1ex]
  & \qquad \qquad X_{(\emptyset, \emptyset)} = 1\, .
\end{array}
\end{equation}
Here we used the notation $(S\cup \{j\}, \alpha \circ u)$ to extend the
assignment $\alpha$ to assign $u$ to the variable indexed by $j$. Note that
the first set of constraints say that the variables should indicate
a consistent assignment.  

Instead of dealing with the constraints of the Sherali-Adams LP relaxation
directly, it is simpler to view each solution of the Sherali-Adams LP as
a consistent collection of local distributions over partial assignments.

Suppose that for every set $S \subseteq [n]$ with 
$|S| \leq r$, we are given a local distribution $\mc{D}(S)$
over $[R]^S$. We say that these distributions are \emph{consistent} 
if for all $S' \subseteq S \subseteq [n]$ with $|S| \leq r$, 
the marginal distribution induced on $[R]^{S'}$ by $\mc{D}(S)$ 
coincides with that of $\mc{D}(S')$. 

The equivalence  between SA solutions and consistent
collections of local distributions basically follows from the definition
of~\eqref{eq:CSP_prog2} and is also used in \cite{CMM} and \cite{CLRS13} that are 
most relevant to our approach. More specifically, we have

\begin{lemma}[Lemma~1 in~\cite{GeorgiouMT09}]
  If $\{\cD(S)\}_{S\subseteq [n]: |S| \leq r}$ is a consistent collection of local  distributions then
  \begin{align*}
    X_{(S, \alpha)} = \Pr_{\cD(S)}[\alpha]
  \end{align*}
  is a feasible solution to~\eqref{eq:CSP_prog2}.
  \label{lem:LoctoSA}
\end{lemma}
Moreover, we have the other direction.
\begin{lemma}
  Consider a feasible solution $(X_{(S, \alpha)})_{S\subseteq [n]: |S| \leq r, \alpha \in [R]^S}$  to~\eqref{eq:CSP_prog2}.
  For each $S \subseteq [n]$ with $|S| \leq r$, define
  \begin{align*}
    \Pr_{\cD(S)}[\alpha] = X_{(S, \alpha)} \qquad \mbox{for each } \alpha \in [R]^S.
  \end{align*}
  Then $(\cD(S))_{S\subseteq [n] : |S| \leq r}$ forms a consistent collection of local distributions.
  \label{lem:SAtoLoc}
\end{lemma}
\begin{proof}
  Note that, for each $S\subseteq n$ with $|S| \leq r$, $\cD(S)$ is indeed a distribution because by the equality constraints of~\eqref{eq:CSP_prog2}
  \begin{align*}
    \sum_{\alpha \in [R]^S}  \Pr_{\cD(S)}[\alpha]  =\sum_{\alpha \in [R]^S} X_{(S, \alpha)} = \sum_{\alpha' \in [R]^{S'}} X_{(S',\alpha')} = X_{(\emptyset, \emptyset)} =1
  \end{align*}
  where $S'\subseteq S$ is arbitrary; and moreover
  $\Pr_{\cD(S)}[\alpha] = X_{(S, \alpha)}\geq 0$.
  Similarly we have, again by the equality constraints of~\eqref{eq:CSP_prog2}, that for each $S' \subseteq S$ and $\alpha' \in [R]^{S'}$  
  \begin{align*}
    \Pr_{\cD(S')}[\alpha']  = X_{(S', \alpha')} = \sum_{\alpha''\in [R]^{S\setminus S'}} X_{(S, \alpha' \circ \alpha'')} =  \sum_{\alpha''\in [R]^{S\setminus S'}} \Pr_{\cD(S)} [\alpha'\circ \alpha'']
  \end{align*}
  so the local distributions are consistent.
\end{proof}

When a SA solution $(X_{(S,\alpha)})$ is viewed as consistent collection $\{\mc{D}(S)\}$ of local distributions, the value of the SA solution can be computed as
$$
\frac{1}{m} \sum_{i=1}^m \sum_{\alpha \in [R]^{S_i}} C_i(\alpha) \cdot X_{(S_i, \alpha)} = \E_{C\in \cC} \left[ \Pr_{\alpha \sim \mc{D}(S_C)}[\alpha \mbox{ satisfies } C] \right]
$$
where $S_C$ is the support of constraint $C$.

%%%%%%%%%%%%%%%%%%%%%%%%%% CSP %%%%%%%%%%%%%%%%%%%%%%%%%% 

\section{Sherali-Adams Integrality Gap for \oFk}
\label{sec-1fcsp}

In this section we establish Sherali-Adams integrality gaps for
\oFk and by virtue of \cite{CLRS13} this extends to general LPs. 
The proof uses the idea of~\cite{CMM} to perform a reduction between problems that
preserves the Sherali-Adams integrality gap.  

Specifically, we show that
the reduction by Bansal and Khot~\cite{BK} from the \uniquegames problem to
$\oFk$ also provides a large Sherali-Adams integrality gap for \oFk, assuming that we start
with a  Sherali-Adams integrality gap instance of \uniquegames.  As large
Sherali-Adams integrality gap  instances of \uniquegames were given
in~\cite{CMM}, this implies the aforementioned integrality gap of $\oFk$.

%%%%%%%%%%%%%%%%%%%%%%%% Unique Games %%%%%%%%%%%%%%%
\subsection{\uniquegames}
\label{UGSec}
The  \uniquegames problem is defined as follows:
\begin{definition}
A \uniquegames instance $\mc{U}=(G,[R],\Pi)$ is defined  by a graph
$G=(V,E)$ over a vertex set $V$ and edge set $E$, where every edge $uv \in E$
is associated with a bijection map $\pi_{u,v} \in \Pi$ such that $\pi_{u,v}:[R]\mapsto [R]$ (we set $\pi_{v,u} := \pi^{-1}_{u,v}$). Here, $[R]$ is
known as the label set. The goal is to find a labeling $\Lambda: V \mapsto [R]$
that maximizes the number of satisfied edges, where an edge $uv$ is
satisfied by $\Lambda$ if $\pi_{u,v}(\Lambda(u))= \Lambda(v)$.
\end{definition}
The following very influential conjecture, known as the \uniquegames
conjecture,  is due to Khot \cite{Khot}.
\begin{conjecture}
\label{conj:UGC}
For any $\zeta, \delta >0$, there exists a sufficiently large constant $R
= R(\zeta,\delta)$ such that the following promise problem is NP-hard. Given a \uniquegames
instance $\mc{U}=(G,[R],\Pi)$, distinguish between the following two
cases: \begin{enumerate}
\item Completeness: There exists a labeling $\Lambda$ that satisfies at least
  $(1-\zeta)$-fraction of the edges.
\item Soundness: No labeling satisfies more than $\delta$-fraction of the edges. 
\end{enumerate}
\end{conjecture}
We remark that the above conjecture has several equivalent formulations via fairly standard transformations. 
In particular, one can assume that the graph $G$ is bipartite and
regular \cite{KR03}. 

The starting point of our reduction is the following Sherali-Adams integrality
gap instances for the \uniquegames problem. Note that  \uniquegames  are
 \constraintsatisfaction and hence here and in the following, we
are concerned with the standard application of the Sherali-Adams hierarchy to CSPs.
\begin{theorem}[\cite{CMM}]
\label{UGgap}
Fix a label size $R = 2^\ell$, a real $\delta \in (0,1)$ and let $\Delta
:= \lceil C(R/\delta)^2 \rceil$ (for a sufficiently large constant C). Then for
every positive $\epsilon$ there exists $\kappa > 0$ depending on $\epsilon$ such
that for infinitely many $n$ there exists an instance of \uniquegames
on a $\Delta$-regular $n$-vertex graph $G=(V,E)$ so that: \begin{enumerate}
\item The value of the optimal solution is at most $\frac{1}{R}\cdot(1+\delta)$.
\item There exists a solution to the LP relaxation obtained after $r
  = n^\kappa$ rounds of the Sherali-Adams relaxation of value $1-\epsilon$.
\end{enumerate}
\end{theorem}

\subsection{Reduction from \uniquegames to $\oFk$}
\label{sec:red}
We first describe the reduction from \uniquegames to $\oFk$ that follows the
construction in~\cite{BK}. We then show that it also preserves the
Sherali-Adams integrality gap.

\paragraph{Reduction} Let $\mc{U}=(G,[R],\Pi)$ be a \uniquegames instance 
over a regular bipartite graph $G=(V,W,E)$. 
Given $\mc{U}$, we construct an instance $\cI$ of
$\oFk$. The reduction has two parameters: $\delta >0$ and $\epsilon >0$, where $\epsilon$ is chosen such 
that $\epsilon R$ is an integer (taking $\epsilon = 2^{-q}$ for some integer $q \geqslant 0$ guarantees this). We then select $t$ to be a large integer depending on $\epsilon$ and 
$\delta$.  

The resulting \oFk instance $\cI$ will be defined over $2^R|W|$ variables and $c|V|$ constraints, where $c:= c(R,\epsilon,t,\Delta)$ is a function of the degree $\Delta$ of the \uniquegames instance, and the constants $R, t$ and $\epsilon$.\footnote{More precisely $c(R,\epsilon,t,\Delta)$ is exponential in the constants $R,t$ and $\epsilon$, and polynomial in $\Delta$ } For our purposes, 
the \uniquegames integrality gap instance that we start from has constant degree 
$\Delta$, and hence $c$ is a constant.

Before we proceed, we stress the fact that our reduction is essentially the same as the one free bit test 
$F_{\epsilon,t}$ in \cite{BK}, but casted in the language of \constraintsatisfaction. The test $F_{\epsilon,t}$ 
expects a labeling $\Lambda: W \mapsto [R]$ for the vertices of the \uniquegames instance, where 
each label $\Lambda(w) \in [R]$ is encoded using a $2^R$ bit string. To check the validity of this labeling, the verifier 
picks a vertex $v\in V$ uniformly at random, and a sequence of $t$ neighbors $w_1,\dots,w_t$ of $v$ randomly and independently from 
the neighborhood of $v$, and asks the provers about the labels of $\{w_1,\dots,w_t\}$ under the labeling $\Lambda$. 
It then accepts if the answers of the provers were convincing, i.e., the labels assigned to $\{w_1,\dots,w_t\}$ satisfy 
the edges $vw_1,\dots,vw_t$ simultaneously under $\pi_{v,w_1},\dots,\pi_{v,w_t}$ respectively. 

Instead of reading all of the $ t2^R$ bits corresponding to the $t$ labels, 
the verifier only reads a \emph{random subset} of roughly $ t 2^{\epsilon R}$ bits and is able to accept with high 
probability if the labeling was correct, and to reject with high probability if it was not correct. In our reduction, the variables 
of the \oFk instance $\cI$ corresponds to the $2^R$ bits encoding the labels of each vertex of the \uniquegames
instance we start from\footnote{For the reader familiar with hardness of approximation and PCP based hardness, we are using the long code to encode labels, 
so that each of these $2^R$ bits gives the value of the dictator 
function $f$ evaluated on a different binary string $x \in \{0,1\}^R$; 
for a valid encoding we have $f(x) = x_\ell$ where $\ell$ is the label 
that is encoded.}, 
and the constraints corresponds to all possible tests that the verifier might perform according to the random 
choice of $v$, the random neighbors $w_1,\dots,w_t$ and the random subset
of bits read by the verifier. Instead of actually enumerating all possible 
constraints, we give a distribution of constraints which is the same as the distribution over the test predicates of $F_{\epsilon,t}$. 

 We refer to the variables
of $\cI$ as follows: it has a binary variable $\langle w, x \rangle$ for each
$w\in W$ and $x\in \{0,1\}^R$.\footnote{$\langle w, x \rangle$ should be
interpreted as the long-code for $\Lambda(w)$ evaluated at $x \in \{0,1\}^R$.}
For further reference, we let
$\textrm{Var}(\cI)$ denote the set of variables of $\cI$. 
The constraints of $\cI$ are picked according to the distribution in Figure~\ref{fig:dist}. 
\begin{figure}[ht]
\begin{mdframed}[linewidth=2]
\begin{enumerate}
\item Pick a vertex $v\in V$ uniformly at random. 
\item Pick $t$ vertices $w_1,\dots,w_t$ randomly and independently from 
the neighborhood $N(v)
  = \{w\in W: vw \in E\}$.
\item Pick $x\in \{0,1\}^R$ at random. \label{s:s3}
\item Let $m = \epsilon R$. Pick indices $i_1,\dots,i_m$ randomly and independently 
from $[R]$ and let $S = \{i_1,\dots,i_m\}$ be the set of those indices. \label{s:s4}
\item Define the \emph{sub-cubes}: \begin{align*}
C_{x,S} & = \{z\in \{0,1\}^R: z_j = x_j \,\,\forall j\notin S \} \\
C_{\bar{x},S}& = \{z\in \{0,1\}^R: z_j = \bar{x}_j\,\,\forall j\notin S \} 
\end{align*} 
\item Output the constraint on the variables $\{\langle w_i,z \rangle \mid i \in [t], \pi_{v,w_i}^{-1}(z) \in C_{x,S} \cup C_{\bar{x},S}\}$ that is true if for some bit $b\in \{0,1\}$ we have
  \begin{align*}
    \begin{array}{rll}
      \langle w_i, z \rangle & = b & \mbox{for all } i \in [t] \mbox{ and } \pi_{v,w_i}^{-1}(z) \in C_{x,S} \text{, and}\\ 
      \langle w_i, z \rangle & = b\oplus 1 & \mbox{for all } i \in [t] \mbox{ and } \pi_{v,w_i}^{-1}(z) \in C_{\bar{x},S} \\ 
    \end{array}
  \end{align*}
  where $\pi(z)$ for $z\in \{0,1\}^R$ is defined as $\pi(z):=(z_{\pi(1)},z_{\pi(2)},\dots,z_{\pi(R)})$, and $\pi^{-1}$ is the inverse map, i.e., $\pi^{-1}(z) \in C_{x,S}$ is equivalent to saying that there exists $y \in C_{x,S}$ such that $\pi(y) = z$. 
\end{enumerate}
\end{mdframed}
\caption{Distribution for the \oFk constraints}
\label{fig:dist}
\end{figure}
%Note that this is exactly the one free bit test $F_{\epsilon,
%t}$ on a \uniquegames instance described in~\cite{BK}. 

It is crucial to observe that our distribution over the constraints
exploits the locality of a \uniquegames solution. To see this, assume
we performed the first two steps of Figure~\ref{fig:dist} and have
thus far fixed a vertex $v\in V$ and $t$ neighbors $w_1,\dots,w_t$,
and let $\mc{C}_{v,w_1,\dots,w_t}$ denote the set of all possible
constraints resulting from steps \ref{s:s3}-\ref{s:s4} (i.e., for all
possible $x\in \{0,1\}^R$ and $S \subseteq [R]$ of size $\epsilon R$).
We will argue that if there exists a local assignment of labels for
$\{v,w_1,\dots,w_t\}$ that satisfies the edges $vw_1,\dots,vw_t$, then
we can derive a local assignment for the variables
$\{\left<w,x\right>: w\in \{w_1,\dots,w_t\} \text{ and } x\in
\{0,1\}^R\}$
that satisfies at least $1-\epsilon$ fraction of the constraints in
$\mc{C}_{v,w_1,\dots,w_t}$. This essentially follows from the
completeness analysis of \cite{BK}, and is formalized in Claim
\ref{claim:completeness}. This allows us to convert a \emph{good}
Sherali-Adams solution of the starting \uniquegames $\mc{U}$, to a
\emph{good} Sherali-Adams solution of the resulting \oFk Instance
$\mc{I}$. Moreover, in order to show that $\mc{I}$ is a Sherali-Adams
integrality gap instance for the \oFk problem, we need to show that
$\textrm{OPT}(\mc{I})$ is \emph{small}. This follows from the
soundness analysis of \cite{BK}, where it was shown that:
\begin{lemma}[soundness]
  \label{lem:soundness}
  For any $\epsilon, \eta > 0$ there exists an integer $t$ so that  $\OPT(\cI) \leq \eta$ if $\OPT(\cU) \leq \delta$ where $\delta>0$ is a constant that only depends on $\epsilon, \eta$ and $t$. 
\end{lemma}
The above says that if we start with a \uniquegames instance $\mc{U}$
with a small optimum then we also get a $\oFk$ instance $\cI$ of small
optimum (assuming that the parameters of the reduction are set
correctly).  In~\cite{BK}, Bansal and Khot also proved the following
completeness: if $\OPT(\mc{U}) \geq 1-\zeta$, then
$\OPT(\cI) \geq 1- \zeta t - \epsilon$.  However, we need the stronger
statement: if $\cU$ has a Sherali-Adams solution of large value, then
so does $\cI$. The following lemma states this more formally, showing
that we can transform a SA solution to the \uniquegames instance $\cU$
into a SA solution to the $\oFk$ instance $\cI$ of roughly the same
value.
\begin{lemma}
\label{fcspSA}
  Let $\{\mu(S) \mid S \subseteq V \cup W, |S| \leq r\}$ be a consistent collection of local distributions defining a solution to the $r$-rounds Sherali-Adams relaxation of the regular bipartite  \uniquegames instance $\cU$. Then we can define a consistent collection of local distributions $\{\sigma(S) \mid S \subseteq  \textrm{Var}(\cI), |S| \leq r\}$ defining a solution to the $r$-rounds Sherali-Adams relaxation of the $\oFk$ instance $\cI$ so that
  \begin{align*}
  \E_{C\in \cC}\left[ \Pr_{\alpha \sim \sigma(S_C)}[ \alpha \mbox{ satisfies } C] \right] \geq (1-\epsilon) \left(1  - t \cdot \E_{vw \in E}\left[  \Pr_{(\Lambda(v), \Lambda(w) \sim \mu(\{v,w\})}[\Lambda(v) \neq \pi_{w,v}(\Lambda(w))] \right] \right),
  \end{align*}
  where $t$ and $\epsilon$ are the parameters of the reduction, and $\sigma(S_C)$ is the distribution over the set of variables in the support $S_C$ of constraint $C$.
\end{lemma}

We remark that the above lemma says that we can transform a SA solution to the \uniquegames instance $\cU$ of value close to 1, into a SA solution to the $\oFk$ instance $\cI$ of value also close to 1. 

\begin{proofof}{Lemma \ref{fcspSA}}
Let $\{\mu (S) \mid S \subseteq V \cup W, |S| \leq r\}$ be  a solution
to the $r$-rounds SA relaxation of the \uniquegames instance $\cU$,
and recall that $\cI$ is the \oFk instance obtained from applying the reduction. We will now use the collection of consistent  local distributions of the \uniquegames instance, to construct another collection of consistent local distributions for the variables in $\textrm{Var}(\cI)$. 

For every set $S\subseteq \textrm{Var}(\cI)$ such that $|S| \leq r$, let $T_S\subseteq W$ be the subset of vertices in the \uniquegames instance defined as follows: \begin{align}
 \label{df1} T_S = \{w \in W: \left<w,x\right> \in S\}.
  \end{align}
   We construct $\sigma(S)$ from $\mu(T_S)$ in the following manner. 
  Given a labeling $\Lambda_{T_S}$ for the vertices in $T_S$ drawn from $\mu(T_S)$, define an assignment $\alpha_S$ for the variables in $S$ as follows: for a variable $\left< w, x \right> \in S$, let $\ell = \Lambda_{T_S}(w)$ be the label of $w$ according to $\Lambda_{T_S}$. Then the new assignment $\alpha_S$ sets $\alpha_S(\left<w,x\right>) := x_\ell$.\footnote{Because $\left<w,x\right>$ is supposed to be the dictator function of the $\ell$th coordinate evaluated at $x$, this is only the correct way to set the bit $\left<w,x\right>$.}
The aforementioned procedure defines a family $\{\sigma(S)\}_{S \subseteq \textrm{Var}(\cI): |S| \leq r}$ of local distributions for the variables of the \oFk instance $\mc{I}$. 
  
  To check that these local distributions are consistent,  take any $S' \subseteq S \subseteq \textrm{Var}(\cI)$ with $|S| \leq r$, and denote by $T_{S'}\subseteq T_{S}$ their corresponding set of vertices as in (\ref{df1}). We know that $\mu(T_S)$ and $\mu(T_{S'})$ agree on $T_{S'}$ since the distributions $\{\mu(S)\}$ defines a feasible Sherali-Adams solution for $\mc{U}$, and hence by our construction, the local distributions $\sigma(S)$ and $\sigma(S')$ agree on $S'$. Combining all of these together, we get that $\{\sigma(S) \mid S \subseteq  \textrm{Var}(\cI), |S| \leq r\}$ defines a feasible solution for the $r$-round Sherali-Adams relaxation of the $\oFk$ instance $\cI$. 
  
  It remains to bound the value of this feasible solution, i.e., \begin{align}
  \label{ex1}
    \E_{C\in \cC}\left[ \Pr_{\alpha \sim \sigma(S_C)}[\alpha \mbox{ satisfies } C] \right].
  \end{align}
   In what follows, we denote by $\psi(.)$ the operator mapping a labeling of the vertices in $T_S$ to an assignment  for the variables in $S$, i.e., $\psi(\Lambda_{T_S}) = \alpha_S$.
  
 First note that a constraint $C \in \mc{C}$ of the \oFk instance $\mc{I}$ is defined by the choice of the vertex $v\in V$, the sequence of $t$ neighbors $\mc{W}_v = (w_1,\dots,w_t)$, the random $x\in \{0,1\}^R$, and the random set $S \subset [R]$ of size $\epsilon R$. We refer to such a constraint $C$ as $C(v,\mc{W}_v,x,S)$. Thus we can rewrite (\ref{ex1}) as \begin{align}
 \label{ex0}
 \E_{v,w_1,\dots,w_t} \left[\Pr_{\Lambda \sim \mu(\{v,w_1,\dots,w_t\}), x ,S} \left[\psi(\Lambda) \mbox{ satisfies } C(v,\mc{W}_v,x,S) \right] \right].
 \end{align}
Recall that the assignment $\psi(\Lambda)$ for the variables $\left\{\left<w,z\right>: w\in \mc{W}_v \text{ and } z\in\{0,1\}^R \right\}$ is derived from the labeling of the vertices in $\mc{W}_v$ according to $\Lambda$. It was shown in \cite{BK} that if $\Lambda$ satisfies the edges $vw_1,\dots,vw_t$ simultaneously, then $\psi(\Lambda)$ satisfies  $C(v,\mc{W}_v,x,S)$ with \emph{high probability}. This is formalized in Claim~\ref{claim:completeness}, whose proof appears in Appendix~\ref{app:proofofclaim}. 

 \begin{claim}
 \label{claim:completeness}
 If $\Lambda$ satisfies $vw_1,\dots,vw_t$ simultaneously, then $\psi(\Lambda)$ satisfies $C(v,\mc{W}_v,x,S)$ with probability at least $1-\epsilon$. Moreover,  if we \emph{additionally} have that  $\Lambda(v) \notin S$, then $\psi(\Lambda)$ always satisfies $C(v,\mc{W}_v,x,S)$.
 \end{claim}

It now follows from Claim~\ref{claim:completeness} that for the assignment $\psi(\Lambda)$ to satisfy the constraint $C(v,\mc{W}_v,x,S)$, it is sufficient that the following two conditions hold \emph{simultaneously}: 
\begin{enumerate}
 \item the labeling $\Lambda$ satisfies the edges $vw_1,\dots,vw_t$;
 \item the label of $v$ according to $\Lambda$ lies outside the set $S$.
\end{enumerate}
 
 Equipped with this, we can use conditioning to lower-bound the probability inside the expectation in (\ref{ex0}) by a product of two probabilities, where the first is \begin{align}
\label{pr1}  \Pr_{\Lambda \sim \mu(\{v,w_1,\dots,w_t\}),x,S} \left[\psi(\Lambda) \mbox{ satisfies } C(v,\mc{W}_v,x,S)  | \Lambda \mbox{ satisfies } vw_1,\dots, vw_t \right]
 \end{align}
 and the second is \begin{align*}
 \Pr_{\Lambda \sim \mu(\{v,w_1,\dots,w_t\})} \left[ \Lambda \mbox{ satisfies } vw_1,\dots, vw_t \right] .
 \end{align*}
 Thus using Claim \ref{claim:completeness}, we get 
\begin{align}
 & \E_{C\in \cC}\left[ \Pr_{\alpha \sim \sigma(S_C)}[\alpha \mbox{ satisfies } C] \right] \geq  (1-\epsilon)\cdot \E_{v,w_1,\dots,w_t} \left[ \Pr_{\Lambda \sim \mu(\{v,w_1,\dots,w_t\})} \left[ \Lambda \mbox{ satisfies } vw_1,\dots, vw_t \right]  \right] \nonumber \\
  \label{ub}   & \geq (1-\epsilon)\left( 1 - \sum_{i=1}^t \E_{v,w_1,\dots,w_t} \left[ \Pr_{\Lambda \sim \mu(\{v,w_1,\dots,w_t\})} \left[ \Lambda \mbox{ does not satisfy } vw_i \right]  \right] \right)\\
  \label{localSat}     & = (1-\epsilon)\left( 1 - \sum_{i=1}^t \E_{v,w_1,\dots,w_t} \left[ \Pr_{\Lambda \sim \mu(\{v,w_i\})} \left[ \Lambda \mbox{ does not satisfy } vw_i \right]  \right] \right)\\
 \label{lasteq}   & = (1-\epsilon) \cdot\left( 1 -  t \cdot \E_{v,w} \left[\Pr_{\Lambda \sim \mu(\{v,w\})} \left[ \Lambda \mbox{ does not satisfy } vw \right] \right]\right)
 \end{align}
 where (\ref{ub}) follows from the union bound, and (\ref{localSat})
 is due to the fact that the local distributions of the \uniquegames
 labeling are consistent, and hence agree on $\{v,w_i\}$. Note that
 the only difference between what we have proved thus far and the
 statement of the lemma, is that the expectation in (\ref{lasteq}) is
 taken over a random vertex $v$ and a random vertex $w\in N(v)$, and
 not random edges. However, our \uniquegames instance we start from is
 regular, so picking a vertex $v$ at random and then a random neighbor
 $w\in N(v)$, is equivalent to picking an edge at random from
 $E$. This concludes the proof.
\end{proofof}

Combining Theorem \ref{UGgap} with Lemmata \ref{lem:soundness} and \ref{fcspSA}, we get the following Corollary.
\begin{corollary}
\label{cor:fcspSA}
For every $\epsilon, \eta>0 $, there exist an arity $k$ and a real $\kappa > 0$ depending on $\epsilon$ and $\eta$ such that for infinitely many $n$ there exists an instance of \oFk of arity $k$ over $n$ variables, so that \begin{enumerate}
\item The value of the optimal solution is at most $\eta$.
\item There exists a solution to the LP relaxation obtained after $r =
  n^{\kappa}$ rounds of the Sherali-Adams relaxation of value at least $1-\epsilon$.
\end{enumerate}
\end{corollary}
\begin{proof}
Let $\mc{U}=(G, [R],\Pi)$ be a $\Delta$-regular \uniquegames instance of Theorem \ref{UGgap} that is $\delta/4$-satisfied with an ${n}^{2 \kappa}_G$-rounds Sherali-Adams solution of value $1-\zeta$, where $n_G$ is the number of vertices in $G$. Note that $G=(V,E)$ is not necessarily bipartite, and our starting instance of the reduction is bipartite. To circumvent this obstacle, we construct a new bipartite \uniquegames instance $\mc{U}'$ from $\mc{U}$ that is $\delta$-satisfied with a Sherali-Adams solution of the same value, i.e., $1-\zeta$. We will later use this new instance to construct our \oFk instance over $n$ variables that satisfies the properties in the statement of the corollary.

 In what follows we think of $\delta,\zeta$ and $R$ as functions of $\epsilon$ and $\eta$, and hence fixing the latter two parameters enables us to fix the constant $t$ of Lemma~\ref{lem:soundness}, and the constant degree $\Delta$ of Theorem \ref{UGgap}. The aforementioned parameters are then sufficient to provide us with the constant arity $k$ of the \oFk instance, along with the number of its corresponding variables and constraints, that is linear in $n_G$.

We now construct the new \uniquegames instance $\mc{U}'$ over a graph $G'=(V_1,V_2,E')$ and the label set $[R]$ from $\mc{U}$ in the following manner: \begin{itemize}
\item Each vertex $v\in V$ in the original graph is represented by two vertices $v_1,v_2$, such that $v_1 \in V_1$ and $v_2 \in V_2$.
\item Each edge $e = uv \in E$ is represented by two edges $e_1 = u_1v_2$ and $e_2=u_2v_1$ in $E'$. The bijection maps $\pi_{u_1,v_2}$ and $\pi_{u_2,v_1}$ are the same as $\pi_{u,v}$.  
\end{itemize}
Note that $G'$ is bipartite  by construction, and since $G$ is $\Delta$-regular, we get that $G'$ is also $\Delta$-regular.  

We claim that no labeling $\Lambda':V_1\cup V_2 \mapsto [R]$ can satisfy more than $\delta$ fraction of the edges in $\mc{U}'$. Indeed, assume towards contradiction that there exists a labeling $\Lambda':V_1\cup V_2 \mapsto [R]$ that satisfies at least $\delta$ fraction of the edges. We will derive a labeling $\Lambda:V \mapsto [R]$ that satisfies at least $\delta/4$ fraction of the edges in $\mc{U}$ as follows: \begin{itemize}
\item[] For every vertex $v\in V$, let $v_1\in V_1$ and $v_2\in V_2$ be its representative vertices in $G'$. Define $\Lambda(v)$ to be either $\Lambda'(v_1)$ or $\Lambda'(v_2)$ with equal probability. 
\end{itemize}
Assume that at least one edge of $e_1=u_1v_2$ and $e_2=u_2v_1$ is satisfied by $\Lambda'$, then the edge $e=uv\in E$ is satisfied with probability at least $1/8$, and hence the expected fraction of satisfied edges in $\mc{U}$ by $\Lambda$ is at least $\delta/4$.

Moreover, we can extend the $r$-rounds Sherali-Adams solution of $\mc{U}$ $\{\mc{D}(S)\}_{S \subseteq V:|S| \leq r}$, to a $r$-rounds Sherali-Adams solution $\{\mc{D}'(S)\}_{S \subseteq V_1 \cup V_2:|S| \leq r}$ for $\mc{U}'$ with the same value. This can be done as follows: For every set $S = S_1 \cup S_2 \subseteq V_1 \cup V_2 $ of size at most $r$, let $S_\mc{U} \subseteq V$ be the set of their corresponding vertices in $G$ and define the local distribution $\mc{D}'(S)$ by mimicking the local distribution $\mc{D}(S_\mc{U})$, repeating labels if the same vertex $v \in S_{\mc{U}}$ has its two copies $v_1$ and $v_2$ in $S$.

Now let $\mc{I}$ be the \oFk instance over $n$ variables obtained by our reduction from the \uniquegames instance $\mc{U}'$, where $n = 2^R n_G$. Since $\OPT(\mc{U}')\leq \delta$, we get from Lemma \ref{lem:soundness} that $\OPT(\mc{I}) \leq \eta$. Similarly, we know from Lemma~\ref{fcspSA} that using an $n_G^{2\kappa}$-rounds Sherali-Adams solution for $\mc{U}'$, we can define an $n^{\kappa}$-rounds Sherali-Adams solution of $\cI$ of roughly the same value, where we used the fact that $R$ is a constant and hence $\left(2^{-2R\kappa} n^{2\kappa} \right)>n^{\kappa}$ for sufficiently large values of $n$. This concludes the proof.
\end{proof}

We have thus far proved that the \oFk problem fools the Sherali-Adams
relaxation even after $n^\kappa$ many rounds for some constant $1>  \kappa > 0$.

%%%%%%%%%%%%%%%%%%% LP-hardness of Vertex Cover and Independent Set %%%%%%%%%%%%%%%%%
\section{LP-hardness of \vertexcover and \independentset }
\label{LPSec}

%%%%%%%%%%%%%%%%% Reductions of LP relaxations %%%%%%%%%
\subsection{Reduction of LP relaxations}
\label{sec:lpreduc}

We will now briefly introduce a formal framework for reducing
between problems that is a stripped down version of the framework 
due to Braun \emph{et al}, with a few notational changes. The 
interested reader can read the details of the full original 
framework in~\cite{BPZ}.

We start with the definition of an optimization problem.

\begin{definition} \label{optProb} An \emph{optimization problem}
  $\Pi=(\mc{S},\mf{I})$ consists of a (finite) set $\mc{S}$ of 
  feasible solutions and a set $\mf{I}$ of instances. 
  Each instance $\mc{I} \in \mf{I}$ specifies an 
  objective function from $\mc{S}$ to $\R_+$. We will denote this 
  objective function by $\Val_{\mc{I}}$ for maximization
  problems, and $\Cost_\mc{I}$ for minimization problems.
  We let $\OPT(\mc{I}) \coloneqq \max_{S \in \mc{S}} \Val_{\mc{I}}(S)$
  for a maximization problem and $\OPT(\mc{I}) \coloneqq
  \min_{S \in \mc{S}} \Cost_{\mc{I}}(S)$ for a minimization 
  problem.
\end{definition}

With this in mind we can give a general definition of the notion 
of an LP relaxation of an optimization problem $\Pi$. We deal
with minimization problems first.

\begin{definition}
  Let $\rho \geq 1$. A \emph{factor-$\rho$ LP relaxation} 
  (or \emph{$\rho$-approximate LP relaxation}) for a 
  minimization problem $\Pi=(\mc{S},\mf{I})$ is a linear system 
  $Ax \geq b$ with $x \in \mathbb{R}^d$ together with the following 
  realizations: 
  \begin{enumerate}[(i)]
\item {\bf Feasible solutions } as vectors $x^S \in \mathbb{R}^d$ for every $S \in \mc{S}$ so that \begin{align*}
Ax^S \geq b && \text{for all } S \in \mc{S}
\end{align*}
\item {\bf Objective functions }via affine functions $f_{\mc{I}} : \mathbb{R}^d \to \mathbb{R}$ for every $\mc{I} \in \mf{I}$ such that \begin{align*}
f_{\mc{I}}(x^S) = \Cost_{\mc{I}}(S) && \text{for all } S \in \mc{S}
\end{align*}
\item {\bf Achieving approximation guarantee $\rho$} via requiring 
\begin{align*}
\OPT(\mc{I}) \leq \rho \LP(\mc{I}) && \text{for all } \mc{I} \in \mf{I}
\end{align*}
where $\LP(\mc{I}) := \min \left\{f_{\mc{I}}(x) \mid Ax \geq b \right\}$.
\end{enumerate}
\end{definition}

Similarly, one can define factor-$\rho$ LP relaxations of a 
maximization problem for $\rho \geq 1$. In our context, the 
concept of a \((c,s)\)-approximate LP relaxation will turn out
to be most useful. Here, $c$ is the \emph{completeness} and 
$s \leq c$ is the \emph{soundness}. For a maximization problem, this 
corresponds to replacing condition (iii) above with
\begin{enumerate}[(i)']
\setcounter{enumi}{2}
\item {\bf Achieving approximation guarantee $(c,s)$} via requiring
\begin{align*}
\OPT(\mc{I}) \leq s \Longrightarrow \LP(\mc{I}) \leq c &&\text{for all } 
\mc{I} \in \mf{I}\,.
\end{align*}
\end{enumerate}

The \emph{size} of an LP relaxation is the number of inequalities 
in $Ax \geq b$. We let $\fc_+(\Pi,\rho)$ denote the minimum size of
a factor-$\rho$ LP relaxation for $\Pi$. In the terminology of~\cite{BPZ},
this is the \emph{$\rho$-approximate LP formulation complexity} of $\Pi$. 
We define \(\fc_+(\Pi,c,s)\) similarly.

In this framework problems can be naturally reduced to each other. We
will use the following restricted form of reductions. 

\begin{definition}
\label{RedDef}
Let $\Pi_1=(\mc{S}_1,\mf{I}_1)$ be a maximization problem and 
$\Pi_2=(\mc{S}_2,\mf{I}_2)$ be a minimization problem. A 
\emph{reduction from $\Pi_1$ to $\Pi_2$} consists 
of two maps, one $\mc{I}_1 \mapsto \mc{I}_2$ from $\mf{I}_1$
to $\mf{I}_2$ and the other $S_1 \mapsto S_2$ from $\mc{S}_1$
to $\mc{S}_2$, subject to 
\begin{align*}
\Val_{\mc{I}_1}(S_1) = \mu_{\mc{I}_1} - \zeta_{\mc{I}_1} \cdot \Cost_{\mc{I}_2}(S_2)
&& \mc{I}_1 \in \mf{I}_1, S_1 \in \mc{S}_1
\end{align*}
where \(\mu_{\mc{I}_1}\) is called the \emph{affine shift} and 
\(\zeta_{\mc{I}_1} \geq 0\) is a normalization factor.

We say that the reduction is \emph{exact} if additionally 
\begin{align*}
\OPT(\mc{I}_1) = \mu_{\mc{I}_1} - \zeta_{\mc{I}_1} \cdot \OPT(\mc{I}_2)
&& \mc{I}_1 \in \mf{I}_1\,.
\end{align*}
\end{definition}

The following result is a special case of a more general result by
\cite{BPZ}. We give a proof for completeness.

\begin{theorem}
\label{BPZmain}
Let \(\Pi_1\) be a maximization problem and let \(\Pi_2\) be a 
minimization problem. Suppose that there exists an exact reduction 
from \(\Pi_1\) to \(\Pi_2\) with \(\mu \coloneqq \mu_{\mc{I}_1}\) 
constant for all \(\mc{I}_1 \in \mf{I}_1\). Then, 
$\fc_+(\Pi_1,c_1,s_1) \leq \fc_+(\Pi_2,\rho_2)$ where 
$\rho_2 \coloneqq \frac{\mu - s_1}{\mu - c_1}$ (assuming
$\mu > c_1 \geq s_1$).
\end{theorem}
\begin{proof}
Let $Ax \geq b$ by a $\rho_2$-approximate LP relaxation for 
$\Pi_2 = (\mc{S}_2,\mf{I}_2)$, with realizations $x^{S_2}$ for 
$S_2 \in \mc{S}_2$ and $f_{\mc{I}_2} : \R^d \to \R$ for $\mc{I}_2
\in \mf{I}_2$. We use the same system $Ax \geq b$ to define a
$(c_1,s_1)$-approximate LP relaxation of the same size for 
$\Pi_1 = (\mc{S}_1,\mf{I}_2)$ by letting 
$x^{S_1} \coloneqq x^{S_2}$ where $S_2$ is the solution of $\Pi_2$
corresponding to $S_1 \in \mc{S}_1$ via the reduction, and 
similarly $f_{\mc{I}_1} \coloneqq \mu - \zeta_{\mc{I}_1} f_{\mc{I}_2}$
with $\zeta_{\mc{I}_1} \geq 0$ where $\mc{I}_2$ is the instance
of $\Pi_2$ to which $\mc{I}_1$ is mapped by the reduction and \(\mu\)
is the affine shift independent of the instance \(\mc{I}_1\).

Then conditions (i) and (ii) of Definition \ref{RedDef}
are automatically satisfied. It suffices to check (iii)' with 
our choice of $\rho_2$, for the given completeness $c_1$ 
and soundness $s_1$. Assume that $\OPT(\mc{I}_1) \leq s_1$
for some instance $\mc{I}_1$ of $\Pi_1$. Then 
\begin{align*}
\LP(\mc{I}_1) 
&= \mu - \zeta_{\mc{I}_1} \LP(\mc{I}_2) 
&& \textrm{(by definition of $f_{\mc{I}_1}$, and since $\zeta_{\mc{I}_1} \geq 0$)}\\
&\leq \mu - \frac{1}{\rho_2} \cdot \zeta_{\mc{I}_1} \cdot \OPT(\mc{I}_2)
&& \textrm{(since $\OPT(\mc{I}_2) \leq \rho_2 \LP(\mc{I}_2)$)}\\
&= \mu + \frac{\mu - c_1}{\mu - s_1} \cdot (\underbrace{\OPT(\mc{I}_1)}_{\leq s_1} - \mu)
&& \textrm{(since the reduction is exact)}\\
&\leq  \mu + \frac{\mu - c_1}{\mu - s_1} \cdot (s_1 - \mu)\\
&= c_1\,,
\end{align*}
as required. Thus $Ax \geqslant b$ gives a $(c_1,s_1)$-approximate
LP relaxation of $\Pi_1$. The theorem follows.
\end{proof}

We will also derive inapproximability of \independentset from a 
reduction between maximization problems. In this case the 
inapproximability factor obtained is of the form 
\(\rho_2 = \frac{\mu + c_1}{\mu + s_1}\).

\subsection{Hardness for \vertexcover and \independentset}
\label{sec:lphardness}
We will now reduce \oFk to \vertexcover with the reduction mechanism
outlined in the previous section, which will yield the desired LP
hardness for the latter problem. 

We start by recasting \vertexcover, \independentset and \oFk
in our language. The two first problems are defined on a fixed
graph $G = (V,E)$.

\begin{problem}[$\vertexcover(G)$]
\label{prob:vc}
  The set of feasible solutions $\mc{S}$ consists of all possible 
  vertex covers $U \subseteq V$, and there is one instance 
  $\mc{I} = \mc{I}(H) \in \mf{I}$ for each induced subgraph $H$ of
  $G$. For each vertex cover \(U\) we have \(\Cost_{\mc{I}(H)}(U) \coloneqq
  |U \cap V(H)|\) being the size of the induced vertex cover in \(H\). 
\end{problem}
  
Note that the instances we consider have 0/1 costs, which makes
our final result stronger: even restricting to 0/1 costs does not
make it easier for LPs to approximate \vertexcover. Similarly, 
for the independent set problem we have:

\begin{problem}[$\independentset(G)$]
\label{prob:indepSet}
  The set of feasible solutions $\mc{S}$ consists
  of all possible independent sets of $G$, and there is one instance 
  $\mc{I} = \mc{I}(H) \in \mf{I}$ for each induced subgraph $H$ of
  $G$. For each independent set \(I \in \mc{S}\), we have that
  \(\Val_{\mc{I}(H)}(I) \coloneqq |I \cap V(H)|\) is the size of 
  the induced independent set of \(H\).
\end{problem}

Finally, we can recast \oFk as follows. Let $n, k \in \N$ be fixed,
with $k \leq n$.

\begin{problem}[$\oFk(n,k)$]
 \label{prob:ofk}
 The set of feasible solutions \(\mc{S}\) consists of all 
 possible variable assignments, i.e., the vertices of the 
 \(n\)-dimensional 0/1 hypercube and there is one instance 
 $\mc{I} = \mc{I}(\mc{P})$ for each possible set $\mc{P} =
 \{P_1, \ldots, P_m\}$ 
 of one free bit predicates of arity $k$. As before, for 
 an instance \(\mc{I} \in \mf{I}\) and an assignment 
 \(x \in \{0,1\}^n\), \(\Val_{\mc{I}}(x)\) is the fraction 
 of predicates $P_i$ that $x$ satisfies (see Definition~\ref{def:1FCSP}). 
\end{problem}

With the notion of LP relaxations and \oFk from above we can now formulate LP-hardness of approximation for \oFk{}s, which   follows directly from Corollary \ref{cor:fcspSA} by
the result of \cite{CLRS13}.

\begin{theorem}
\label{mainfcsp}
For every $\epsilon>0$ there exists a constant arity $k = k(\epsilon)$ 
such that for infinitely many $n$ we have \(\fc_+(\oFk(n,k),1-\epsilon,\epsilon) \geq
n^{\Omega \left(\log n / \log \log n\right)}\).
\end{theorem}

Following the approach in \cite{BPZ}, we define a graph \(G\)
over which we consider \vertexcover, which will correspond to our
(family of) hard instances. This graph is a \emph{universal} FGLSS
graph as it encodes all possible choices of predicates
simultaneously \cite{feige1991approximating}. The constructed graph 
is similar to the one in \cite{BPZ}, however now we consider 
\emph{all} one free bit predicates and not just the \maxcut 
predicate \(x \oplus y\).

\begin{definition}[\vertexcover host graph]
  For fixed number of variables $n$ and arity $k \leq n$ 
  we define a graph \(G^* = G^*(n,k)\) as follows. Let
  $x_1$, \ldots, $x_n$ denote the variables of the CSP.

  \emph{Vertices:} For every one free bit predicate \(P\) of arity
  \(k\) and subset of indices \(S \subseteq [n]\) of size \(k\) we 
  have two vertices \(v_{P,S, 1}\) and \(v_{P,S, 2}\) corresponding
  to the two satisfying partial assignments for \(P\) on variables
  $x_i$ with $i \in S$. For simplicity we identify the partial 
  assignments with 
  the respective vertices in \(G^*\). Thus a partial assignment 
  \(\alpha \in \{0,1\}^{S}\) satisfying predicate \(P\) has a 
  corresponding vertex \(v_{P,\alpha} \in \{v_{P,S,1},v_{P,S,2}\}\).

\emph{Edges:} Two vertices \(v_{P,\alpha_1}\) and
\(v_{P,\alpha_2}\) are connected if and only if the corresponding
partial assignments $\alpha_1$ and $\alpha_2$ are
incompatible, i.e., there exists \(i \in S_1 \cap S_2\) with
\(\alpha_1(i) \neq \alpha_2(i)\). 
\end{definition}

Note that the graph has \(2 \binom{2^k}{2} \binom{n}{k}\) vertices, which
is polynomial in \(n\) for fixed \(k\). In order to establish LP-inapproximability 
of \vertexcover and \independentset it now suffices to define a reduction 
satisfying Theorem~\ref{BPZmain}.

\begin{maintheorem}
  For every $\epsilon>0$ and for infinitely many $n$, there exists a
  graph $G$ with $|V(G)|=n$ such that
  $\fc_+(\vertexcover(G),2-\epsilon) \geq n^{\Omega\left( \log n /\log
        \log n \right)}$, and also
  $\fc_+(\independentset(G),1/\epsilon) \geq n^{\Omega\left( \log n /\log
        \log n \right)}$.
\end{maintheorem}
\begin{proof}
We reduce \oFk on $n$ variables with sufficiently large arity $k = k(\epsilon)$
to \vertexcover over $G \coloneqq G^*(n,k)$. 
For a \oFk instance $\mc{I}_1 \coloneqq \mc{I}_1(\mc{P})$ and set of
predicates $\mc{P} = \{P_1,\ldots,P_m\}$, let $H(\mc{P})$ be the induced 
subgraph of $G$ on the set of vertices
$V(\mc{P})$ corresponding to the partial assignments satisfying 
some constraint in $\mc{P}$. So $V(\mc{P}) = \{v_{P,S,i} \mid P \in \mc{P}, 
S \subseteq [n], |S| \leq k, i = 1, 2\}$. 

In Theorem~\ref{mainfcsp} we have shown that no LP of size
at most $n^{o\left(\log n / \log \log n\right)}$ can
provide an $(1-\epsilon,\epsilon)$-approximation for \oFk for any
$\epsilon>0$, provided the arity $k$ is large enough. To prove 
that every LP relaxation with $2-\epsilon$ approximation guarantee 
for \vertexcover has size at least $n^{\Omega\left(\log n / \log 
\log n \right)}$, we provide maps defining a reduction from \oFk 
to \vertexcover. 
   
In the following, let \(\Pi_1 = (\mc{S}_1, \mf{I}_1)\) be the 
\oFk problem and let \(\Pi_2 = (\mc{S}_2, \mf{I}_2)\) be the 
\vertexcover problem. In view of Definition~\ref{RedDef}, we 
map $\mc{I}_1 = \mc{I}_1(\mc{P})$ to $\mc{I}_2 = \mc{I}_2(H(\mc{P}))$
and let \(\mu \coloneqq 2\) and \(\zeta_{\mc{I}_1} \coloneqq \frac{1}{m}\) 
where \(m\) is the number of constraints in \(\mc{P}\).

For a total assignment $x \in \mc{S}_1$ we define
$U = U(x) \coloneqq \{ v_{P,\alpha} \mid \alpha$ satisfies
$P$ and $x$ does not extend $\alpha\}$. The latter
is indeed a vertex cover: we only have edges between conflicting
partial assignments, and all the partial assignments that agree with
$x$ are compatible with each other. Thus $I = I(x) \coloneqq
\{ v_{P,\alpha} \mid \alpha$ satisfies $P$ and $x$ extends $\alpha\}$
is an independent set and its complement \(U\) is a vertex cover. 

We first verify the condition that \(\Val_{\mc{I}_1}(x) = 2 - \frac{1}{m}
\Cost_{\mc{I}_2}(U(x))\) for all instances \(\mc{I}_1 \in \mf{I}_1\) and 
assignments \(x \in \mc{S}_1\). Every predicate $P$ in $\mc{P}$ over 
the variables in $\{x_{i} \mid i \in S\}$ has exactly two 
representative vertices $v_{P,\alpha_1}$, $v_{P,\alpha_2}$ where the
$\alpha_1, \alpha_2 \in \{0,1\}^{S}$ are the two partial assignments satisfying $P$. 
If an assignment $x \in \mc{S}_1$ satisfies the predicate $P$, then exactly 
one of $\alpha_1, \alpha_2$ is compatible with $x$. Otherwise, when $P(x) = 0$,
neither of $\alpha_1,\alpha_2$ do. This means that in the former case 
exactly one of \(v_{P,\alpha_1},v_{P,\alpha_2}\) is contained in \(U\) and in 
the latter both \(v_{P,\alpha_1}\) and \(v_{P,\alpha_2}\) are contained in \(U\). 
It follows that for any \(\mc{I}_1 = \mc{I}_1(\mc{P}) \in \mf{I}_1\) and \(x \in \mc{S}_1\) it holds
\[
\Val_{\mc{I}_1}(x) = 2 - \frac{1}{m} \Cost_{\mc{I}_2}(U(x))\,.
\] 
In other words, for any specific \(\mc{P}\) the affine shift is \(2\),
and the normalization factor is $\frac{1}{m}$. 

Next we verify exactness of the reduction, i.e.,
\[\OPT(\mc{I}_1) = 2 - \frac{1}{m} \OPT(\mc{I}_2)\,.\]
For this take an arbitrary vertex cover $U \in \mc{S}_2$ 
of $G$ and consider its complement. This is an independent 
set, say $I$. As \(I\) is an independent set, all partial assignments $\alpha$
such that $v_{P,\alpha} \in I$ are compatible and there exists a 
total assignment $x$ that is compatible with each $\alpha$ with 
$v_{P,\alpha} \in I$. Then the corresponding vertex cover $U(x)$
is contained in $U$. Thus there always exists an optimum solution 
to $\mc{I}_2$ that is of the form $U(x)$. Therefore, the reduction 
is exact.

It remains to compute the inapproximability factor via
Theorem~\ref{BPZmain}. 
We have 
\begin{align*}
  \rho_2 =  \frac{2 -
  \epsilon}{2 - (1 - \epsilon)}
\geq 2 - 3\epsilon
\end{align*}

A similar reduction works for \independentset. This time, the
affine shift is $\mu = 0$ and we get an inapproximability factor of
$$
\rho_2 = \frac{1-\epsilon}{\epsilon} \geq \frac{1}{2\epsilon}
$$
for $\epsilon$ small enough.
\end{proof}

\section{Upper bounds}
\label{sec:upperbounds}

Here we give a size-$O(n)$ LP relaxation for approximating \independentset within 
a factor-$O(\sqrt{n})$, which follows directly by work of Feige and
Jozeph~\cite{FJ14}. Note that this is strictly better than the
\(n^{1-\epsilon}\) hardness obtained assuming \(P \neq NP\) by \cite{haastad1996clique}. This
is possible because the \emph{construction} of our LP is NP-hard while
being still of small size, which is
allowed in our framework. 

Start with a greedy coloring of $G = (V,E)$: let $I_1$ be any maximum size 
independent set of $G$, let $I_2$ be any maximum independent set of $G - I_1$, 
and so on. In general, $I_{j+1}$ is any maximum independent set of $G - I_1 - 
\cdots - I_j$. Stop as soon as $I_1 \cup \cdots \cup I_j$ covers the whole 
vertex set. Let $k \leq n$ denote the number of independent sets constructed,
that is, the number of colors in the greedy coloring. 

Feige and Jozeph~\cite{FJ14} made the following observation:

\begin{lemma} \label{lem:greedy1}
Every independent set $I$ of $G$ has a nonempty intersection with at most 
$\lfloor 2 \sqrt{n} \rfloor$ of the color classes $I_j$.
\end{lemma}

Now consider the following linear constraints in $\R^V \times \R^k \simeq \R^{n+k}$:
\begin{align}
\label{eq:greedy1} &0 \leq x_v \leq y_j \leq 1 \qquad \forall j \in [k], v \in I_j\\
\label{eq:greedy2} &\sum_{j = 1}^k y_j \leq \lfloor 2 \sqrt{n} \rfloor\,.
\end{align}
These constraints describe the feasible set of our LP for \independentset on $G$. 
Each independent set $I$ of $G$ is realized by a 0/1-vector $(x^I,y^I)$ defined 
by $x^I_v = 1$ iff $I$ contains vertex $v$ and $y^I_j = 1$ iff $I$ has a 
nonempty intersection with color class $I_j$. For an induced subgraph $H$ 
of $G$, we let $f_{\mc{I}(H)}(x,y) \coloneqq \sum_{v \in V(H)} x_v$. 
By Lemma~\ref{lem:greedy1}, $(x^I,y^I)$ satisfies 
\eqref{eq:greedy1}--\eqref{eq:greedy2}. 
Moreover, we clearly have $f_{\mc{I}(H)}(x^I,y^I) 
= |I \cap V(H)|$. Let $\LP(\mc{I}(H)) \coloneqq \max \{f_{\mc{I}(H)}(x,y) \mid 
\eqref{eq:greedy1}, \eqref{eq:greedy2}\} = \max \{\sum_{v \in V(H)} x_v \mid \eqref{eq:greedy1}, \eqref{eq:greedy2}\}$.

\begin{lemma} \label{lem:greedy2}
For every induced subgraph $H$ of $G$, we have
$$
\LP(\mc{I}(H)) \leqslant \lfloor 2 \sqrt{n} \rfloor \OPT(\mc{I(H)})\,.
$$
\end{lemma}
\begin{proof}
When solving the LP, we may assume $x_v = y_j$ for all $j \in [k]$ and all 
$v \in I_j$. Thus the LP can be rewritten
$$
\max \left\{\sum_{j = 1}^k |I_j \cap V(H)| \cdot y_j \mid 0 \leq y_j \leq 1\ 
\forall j \in [k],\ \sum_{j = 1}^k y_j \leq \lfloor 2 \sqrt{n} \rfloor\right\}\,.
$$
Because the feasible set is a 0/1-polytope, we see that the optimum value of 
this LP is attained by letting $y_j = 1$ for at most $\lfloor 2 \sqrt{n} \rfloor$
of the color classes $I_j$ and $y_j = 0$ for the others. Thus some color 
class $I_j$ has weight at least $1/\lfloor 2 \sqrt{n} \rfloor$ of the LP 
value.
\end{proof}

By Lemma~\ref{lem:greedy2}, constraints 
\eqref{eq:greedy1}--\eqref{eq:greedy2} provide a size-$O(n)$
factor-$O(\sqrt{n})$ LP relaxation of \independentset.

\begin{theorem}
For every $n$-vertex graph $G$, $\fc_+(\independentset(G),2\sqrt{n})
\leqslant O(n)$.
\end{theorem}

Although the LP relaxation 
\eqref{eq:greedy1}--\eqref{eq:greedy2}
is NP-hard to construct, it is allowed by our 
framework because we do not bound the time
needed to construct the LP. To our knowledge, this is the first example of 
a polynomial-size extended formulation outperforming polynomial-time 
algorithms.

We point out that a factor-$n^{1-\epsilon}$ LP-inapproximability of 
\independentset holds in a different model, known as the \emph{uniform 
model}~\cite{BM13,BP13}. In that model, we seek an LP relaxation that 
approximates \emph{all} \independentset instances with the same number of 
vertices $n$. This roughly corresponds to solving \independentset 
by approximating the correlation polytope in some way, which turns
out to be strictly harder than approximating the stable set polytope,
as shown by our result above.

\section{LP Hardness for \ekvertexcover}
\label{sec:qvertexcover}
\newcommand{\uCSP}{\problemmacro{Not-Equal-CSP}}
\newcommand{\Tf}[1]{\Upsilon_{#1}}
\newcommand{\sconstraintsatisfaction}{\problemmacro{constraint satisfaction problem}}
In order to prove LP lower bounds for \vertexcover and \independentset in Sections~\ref{sec-1fcsp} and~\ref{LPSec}, we first started by providing a reduction from the \uniquegames problem to the \oFk problem, that implied that no \emph{small size} linear program is a ($1-\epsilon,\epsilon$)-approximation for the \oFk problem. We then gave a gap-preserving reduction from any LP approximating $\oFk$ to any LP approximating \vertexcover, and showed that no small size LP can provide a ($2-\epsilon$)-approximation for the \vertexcover problem.

Our approach for the \ekvertexcover will be similar, however our starting point is a \sconstraintsatisfaction different than \oFk. This new \sconstraintsatisfaction, that we refer to as \uCSP, is defined as follows: 

\begin{definition}
\label{def:ucsp}
  A CSP of arity $k$ over the domain\footnote{For convenience, we use the additive
  group $\Z_q = \{0,\ldots,q-1\}$ instead of $[q]$ as the domain of our CSP.} 
  $\Z_q$ is referred to as \emph{\uCSP} if each  constraint 
  $P : \Z_q^k \rightarrow \{0,1\}$ is of the following form
  \begin{align*}
    P_A(x_1, x_2, \ldots, x_k) = 1 \qquad \mbox{if and only if} \qquad \bigwedge_{i=1}^k (x_i \neq a_i) 
  \end{align*}
  for some $A = (a_1, a_2, \ldots, a_k) \in \Z_q^k$. When $x\in \Z_q^n$, for some $n \geq k$, a predicate $P:=P_{S,A}$ is additionally indexed by a set $S = \{ i_1,i_2,\dots,i_k\}\subseteq [n]$, and $P_{S,A}(x) = P_A(x_{i_1},x_{i_2},\dots,x_{i_k})$.
\end{definition}

We remark that the above definition should not be confused with the common
Not-\emph{All}-Equal predicate.

Similar to the approach of \vertexcover, we shall prove that there is no
small linear programming relaxation for \ekvertexcover with a good
approximation guarantee in two steps. In the first step, we prove that no small
linear programming relaxation can \emph{approximate well} the \uCSP problem. We then give a gap-preserving reduction from this problem to that of \ekvertexcover in the framework of \cite{BPZ}.

\subsection{Sherali-Adams Integrality Gap for \uCSP}
This section will be dedicated to proving the following theorem. 
\begin{theorem}
\label{thm:mainUCSP}
  For any $\varepsilon > 0$ and integer $q\geq 2$, there exist  $\kappa >0$ and an integer $k$ so that for infinitely many $n$ there exists a \uCSP{} instance $\cI$ of arity $k$ over $n$ variables satisfying
  \begin{itemize}
    \item $\OPT(\cI) \leq \varepsilon$;
    \item There is a solution to the $n^\kappa$-round Sherali-Adams relaxation  of value $1-1/q - \varepsilon$.
  \end{itemize}
  \label{}
\end{theorem}

The above theorem states that the \uCSP problem can fool the Sherali-Adams relaxation even after $n^{\kappa}$ many rounds. Before we proceed, we discuss functions of the form $f:\Z_q^R\mapsto \{0,1\}$. These functions will play a crucial role in the analysis.

%%%%%%%%%%%%%%%%%%%%% q-ary Functions %%%%%%%%%%%%%%%%%%%%%%%%

\subsubsection{Functions Over the Domain $\Z_q$}
In order to construct Sherali-Adams integrality gaps for the \uCSP problem, we also reduce from the \uniquegames problem. The analysis of this reduction relies heavily on known properties regarding functions of the form $f: \Z_q^R \mapsto \{0,1\}$, where $\Z_q$ is to be thought of as the domain of the new CSP, and $R$ as the label set size of the \uniquegames instance. More precisely, we exploit the drastic difference in the behavior of functions depending on whether they have \emph{influential} coordinates or not. To quantify these differences, we first need the following definitions.
\begin{definition}
For a function $f: \Z_q^R \mapsto \{0,1\}$, and an index $i \in [R]$, the influence of the $i$-th coordinate is given by
 \begin{align*}
\Infl{i}{}{f} = \E \left[{\bf Var}\left[f(x) \large | x_1,\dots,x_{i-1},x_{i+1},\dots,x_n \right] \right]
\end{align*}
where $x_1,\dots,x_n$ are uniformly distributed.
\end{definition}
An alternative definition for the influence requires defining the Fourier expansion of a function $f$ of the form $f: \Z_q^R \mapsto \{0,1\}$. To do this, let $\phi_0 \equiv 1, \phi_1,\dots,\phi_{q-1}: \Z_q \mapsto \R$ be such that for all $i, j \in [q]$, we have 
\begin{align*}
\E_{y \in \Z_q} \left[ \phi_i(y) \phi_j(y)\right] = 
\begin{cases}
0 &\text{if } i \neq j\\
1 &\text{if } i = j
\end{cases}
\end{align*}
where the expectation is taken over the uniform distribution, and define the functions $\phi_\alpha: \Z_q^R \mapsto \R$ for every $\alpha \in \Z_q^R$ to be  
\begin{align*}
\phi_\alpha(x) := \prod_{i=1}^R \phi_{\alpha_i} \left(x_i\right)
\end{align*}
for any $x \in \Z_q^R$. We take these functions for defining our Fourier basis. Note that this coincides with the boolean case, where for $b\in \{0,1\}$ we have $\phi_0(b) \equiv 1$, and $\phi_1(b) = (-1)^b$ (or the identity function in the $\{-1,1\}$ domain). For a more elaborate discussion on the Fourier expansion in generalized domains, we refer the interested reader to Chapter 8 in \cite{RyanBook}. 

Having fixed the functions $\phi_0, \phi_1,\dots,\phi_{q-1}$, every function $f: \Z_q^R \mapsto \{0,1\}$ can be uniquely expressed as 
\begin{align*}
f (x)= \sum_{\alpha \in \Z_q^R} \hat{f}_\alpha \phi_\alpha(x)
\end{align*}
Equipped with this, we can relate the influence of a variable $i\in [R]$ with respect to a function $f: \Z_q^R \mapsto \{0,1\}$, to the Fourier coefficients of $f$ as follows:  
\begin{align*}
\Infl{i}{}{f} = \sum_{\alpha: \alpha_i\neq 0} \hat{f}^2_\alpha
\end{align*}
%
%\begin{comment}
%\begin{definition}
%For a function $f: \Z_q^R \mapsto \{0,1\}$, and an index $i \in [R]$, the influence of the $i$-th coordinate is given by
%%
%\begin{align*}
%\Infl{i}{}{f} = \sum_{\alpha: \alpha_i\neq 0} \hat{f}^2_\alpha
%\end{align*}
%%
%%where  $f = \sum_{\phi \in \Z_q^R} \hat{f}(\phi) \chi_\phi$ is the multi-linear representation of $f$.
%where $\hat{f}(\phi) = \E_{x\in \Z_q^R}\left[f(x) \chi_\phi(x)\right]$ and $\chi_\phi(x) = \prod_{i=1}^R{\omega^{\phi_i x_i}}$ for a $q^{th}$ root on unity $\omega$.\footnote{In analogy with the standard Fourier representation, the characters $\{\chi_\phi\}_{\phi \in \Z_q^R}$ define an orthonormal basis of the vector space of all functions $f:\Z_q^R \mapsto \mathbb{R}$.}
%% \redTodo{Sam: what are those functions exactly? Shouldn't we say it or give a reference?}}
%\end{definition}
%\end{comment}
%
In our analysis we will however be interested in \emph{degree-d} influences, denoted $\Infl{i}{d}{d}$ and defined as \begin{align*}
\Infl{i}{d}{f} = \sum_{\alpha: \alpha_i \neq 0, |\alpha| \leq d} \hat{f}^2_\alpha
\end{align*}
where $|\alpha|$ in this context is the support of \(\alpha\), i.e., the number of indices $j\in [R]$ such that $\alpha_j \neq 0$.

\begin{observation}[see, e.g., Proposition 3.8 in \cite{MOO05}]
\label{obs:dinfl}
For a function $f:\Z_q^R \mapsto \{0,1\}$, the sum of all degree-$d$ influences is at most $d$.
\end{observation}

We will also need a generalization of the notion of sub-cubes defined in Figure~\ref{fig:dist} in order to state the ''It Ain't Over Till It's Over'' Theorem \cite{MOO05}, a main ingredient of the analysis of the reduction. In fact we only state and use a special case of it, as it appears in \cite{O12}.

\begin{definition}
Fix $\epsilon>0$. For $x\in \Z_q^R$, and $S_\epsilon \subseteq[R]$ such that $|S_\epsilon| = \epsilon R$, the sub-cube $C_{x,S_\epsilon}$ is defined as follows: \begin{align*}
C_{x,S_{\epsilon}}:= \left\{z \in \Z_q^R: z_j = x_j \,\, \forall j \notin S_\epsilon \right\}
\end{align*}
\end{definition}

\begin{theorem}[\emph{Special case of the} {\normalfont It Ain't Over Till It's Over} \emph{Theorem}]
For every $\epsilon, \delta>0$ and integer $q$, there exist $\vartheta > 0$ and integers $t, d$ such  that any collection of functions $f_1,\dots,f_t: \Z_q^R \mapsto \{0,1\}$ that satisfies \begin{align*}
\forall j: \E\left[f_j\right] \geq \delta && \text{ and } && \forall i \in [R], \forall 1 \leq \ell_1\neq \ell_2 \leq t: \min \left\{\text{\normalfont Inf}_i^d(f_{\ell_1}), \text{\normalfont Inf}_i^d(f_{\ell_2})\right\} \leq \vartheta,
\end{align*}
has the property
 \begin{align*}
\Pr_{x,S_\epsilon} \left[ \bigwedge_{j=1}^t (f_j(C_{x,S_\epsilon})\equiv 0)\right] \leq \delta.
\end{align*}
\label{thm:over}
\end{theorem}

Essentially what this theorem says is that if a collection of $t$ fairly balanced functions are all identical to zero on the same random sub-cube with non-negligible probability, then at least two of these functions must share a common influential coordinate. In fact all the functions that we  use throughout this section   satisfy a strong balance property, that we denote by \emph{folding}.\footnote{We abuse the notion of folding here, and we stress that this should not be confused with the usual notion of folding in the literature, although it coincides with standard folding for the boolean case.}

\paragraph{Folded Functions}
We say that a function $f: \Z_q^R \mapsto \{0,1\}$ is \emph{folded} if every
\emph{line} of the form $\{x \in \Z_q^R \mid x = a + \lambda \mathbf{1}, \lambda 
\in \Z_q\}$ contains a unique point where $f(x)$ is zero, where $\mathbf{1} \in 
\Z_q^R$ is the all-one vector and $a \in \Z_q^R$ is any point.

\begin{remark}
\label{rem:expec}
For any folded function $f: \Z_q^R \mapsto \{0,1\}$, we have that  $\E_x \left[f (x) \right] = 1 - 1/q$. %The expected value of any folded function $f: \Z_q^R \mapsto \{0,1\}$ is $\E_x \left[f (x) \right] = 1 - 1/q$. 
\end{remark}
We shall also extend the notion of \emph{dictatorship} functions restricted to the folded setting. In this setting, the $\ell$-th coordinate dictator function  $f_\ell: \Z_q^R \mapsto \{0,1\}$ for some $\ell \in [R]$ is defined as
\begin{align*}
f_\ell(x ) = \left \{ \begin{array} {l l}
1 & \text{ if $x_\ell \neq 0$ }\\
0 & \text{ if $x_\ell = 0$\,.}
\end{array} \right.
\end{align*}
Notice that $f_\ell$ is folded because it is zero exactly on the coordinate hyperplane $\{x \in \Z_q^R \mid x_\ell = 0\}$.
 
\paragraph{Truth Table Model}
In order to guarantee the folding property of a function $f: \Z_q^R \mapsto \{0,1\}$ in the truth table model, we adopt the following convention:\begin{itemize}
\item The truth table $\Tf{f}$ has $q^{R-1}$ entries in $\Z_q$, one for each 
$x \in \Z_q^R$ such that $x_1 = 0$.
\item For each $x \in \Z_q^R$ with $x_1 = 0$, the corresponding entry $\Tf{f}(x)$
contains the unique $\lambda \in \Z_q$ such that $f(x + \lambda \mathbf{1}) = 0$.
\end{itemize}
We can however use $\Tf{f}$ to query $f(x)$ for any $x \in \Z_q^R$ as follows: we have $f(x) = 0$ whenever $\Tf{f}{(x-x_1 \mathbf{1})} = \Tf{f}(0,x_2-x_1,\ldots,x_R-x_1) = x_1$ and $f(x) = 1$ otherwise.

We can now readily extend the notion of the  \emph{long code} encoding to match our definition of dictatorship functions.
\begin{definition}
\label{def:LC}
The long code encoding of an index $\ell \in [R]$ is simply $\Tf{f_\ell}$, the truth table of the \emph{folded} dictatorship function of the $\ell$-th coordinate. Similarly, the long code $\Tf{f_\ell} \in \Z_q^{q^{R-1}}$ is indexed by all $x\in \Z_q^R$ such that $x_1 = 0$. 
\end{definition} 
%\begin{remark}
%\label{rem:dict}
%Given our definition of $f_\ell$, $\Tf{f_\ell}(x)$ can be computed for any $x\in \Z_q^R$ with $x_0 = 0$ to be $\Tf{f_{\ell}}(x):= (q-x_{\ell})\mod q$, i.e., the additive inverse of $x_\ell$ modulo $q$.
%\end{remark}
%%%%%%%%%%%%%%%%% SA Gaps %%%%%%%%%%%%%%%%%%%%%%

\subsubsection{Reduction from \uniquegames to \uCSP}
We first describe the reduction from \uniquegames to $\uCSP$ that is similar in many aspects to the reduction in Section~\ref{sec:red}. We then show that it also preserves the
Sherali-Adams integrality gap.

\paragraph{Reduction} Let $\mc{U}=(G,[R],\Pi)$ be a \uniquegames instance 
over a regular bipartite graph $G=(V,W,E)$. 
Given $\mc{U}$, we construct an instance $\cI$ of
$\uCSP$. The reduction has three parameters: an integer $q\geq2$ and reals $\delta, \epsilon >0$, where $\epsilon$ is chosen such 
that $\epsilon R$ is an integer. We then select $t$ to be a large integer depending on $\epsilon, \delta$ and $q$ so as to satisfy Lemma~\ref{lem:qsound}.  

The resulting \uCSP instance $\cI$ will be defined over $|W|q^{R-1}$ variables and $c|V|$ constraints, where $c:= c(R,\epsilon,t,\Delta,q)$ is a function of the degree $\Delta$ of the \uniquegames instance, and the constants $R, t, q$ and $\epsilon$. For our purposes, 
the \uniquegames integrality gap instance that we start from, has constant degree 
$\Delta$, and hence $c$ is a constant.

 We refer to the variables
of $\cI$ as follows: it has a variable $\langle w, z \rangle \in \Z_q$ for each
$w\in W$ and $z \in \Z_q^R$ such that $z_1 = 0$.
For further reference, we let
$\textrm{Var}(\cI)$ denote the set of variables of $\cI$. 
The constraints of $\cI$ are picked according the distribution in Figure~\ref{fig:dist2} on page~\pageref{fig:dist2}. One can see that a constraint $C:= C(v,\mc{W}_v,x,S_\epsilon)$ is then defined by the random vertex $v$ (Line~\ref{line:l1}), the $t$ random neighbors $\mc{W}_v = \{w_1,\dots,w_t\}$ (Line~\ref{line:l2}), the random $x\in \Z_q^R$ (Line~\ref{line:l3}) and the random subset $S_\epsilon \subseteq[R]$ (Line~\ref{line:l4}).
 \begin{figure}[ht]
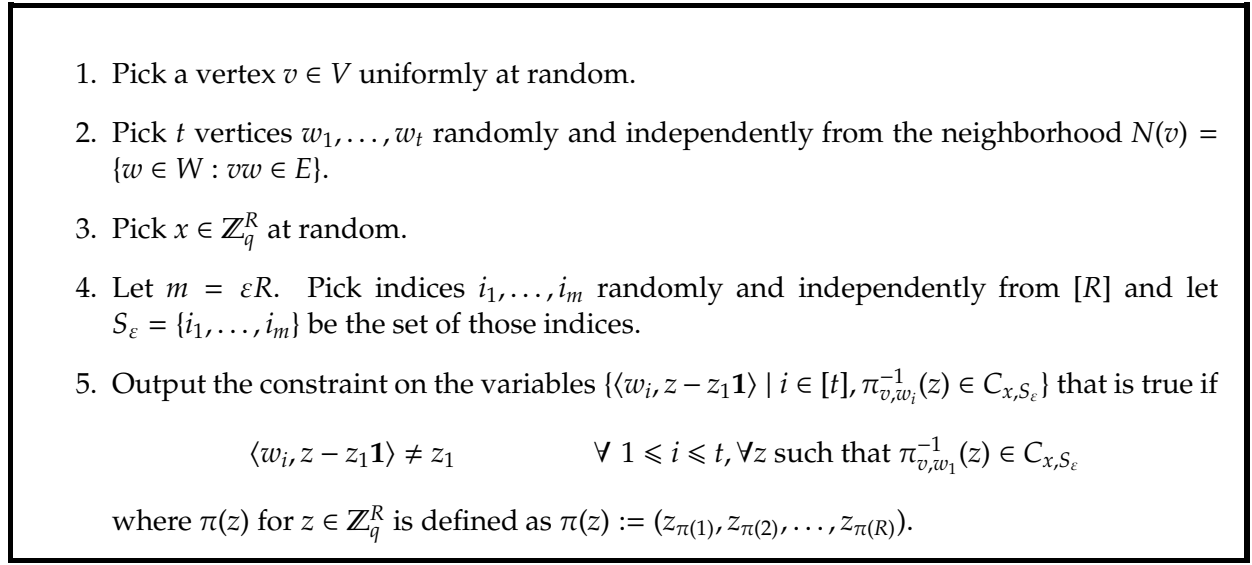

\begin{mdframed}[linewidth=2]
\begin{enumerate}
\item Pick a vertex $v\in V$ uniformly at random.  \label{line:l1}
\item Pick $t$ vertices $w_1,\dots,w_t$ randomly and independently from 
the neighborhood $N(v)
  = \{w\in W: vw \in E\}$.
\label{line:l2}
\item Pick $x\in \Z_q^R$ at random. \label{line:l3}
\item Let $m = \epsilon R$. Pick indices $i_1,\dots,i_m$ randomly and independently \label{line:l4} 
from $[R]$ and let $S_\epsilon = \{i_1,\dots,i_m\}$ be the set of those indices. 
\item Output the constraint on the variables $\{\langle w_i,z-z_1 \mathbf{1} \rangle \mid i \in [t], \pi_{v,w_i}^{-1}(z) \in C_{x,S_\epsilon} \}$ that is true if   \begin{align*}
\left<w_i,z - z_1 \mathbf{1}\right> \neq z_1 && \forall \,\, 1\leq i \leq t, \forall z \text{ such that }\pi^{-1}_{v,w_1}(z) \in C_{x,S_\epsilon}
\end{align*}
  where $\pi(z)$ for $z\in \Z_q^R$ is defined as $\pi(z):=(z_{\pi(1)},z_{\pi(2)},\dots,z_{\pi(R)})$.
\end{enumerate}
\end{mdframed}
\caption{Distribution for the \uCSP constraints}
\label{fig:dist2}
\end{figure}

Note that if we think of the variables $\left<w,z\right>$ for a fixed $w\in W$ as the truth table of some function $f_w:\Z_q^R \mapsto \{ 0,1\}$, then $f$ is \emph{forced} to satisfy the folding property. 

We claim that if the starting \uniquegames instance $\mc{U}$ was a Sherali-Adams integrality gap instance, then $\mc{I}$ is also an integrality gap instance for the \uCSP problem. Similar to Section~\ref{sec:red}, we prove this in two steps; we first show that if $\OPT(\mc{U})$ is \emph{small}, then so is $\OPT(\mc{I})$. Formally speaking, the following holds:
\begin{lemma}
\label{lem:qsound}
For every $\epsilon, \eta >0$ and alphabet size $q\geq2$ there exists an integer $t$ so that $\OPT(\mc{I}) \leq \eta$ if $\OPT(\mc{U}) \leq \delta$ where $\delta>0$ is a constant that only depends on $\epsilon, \eta, q$ and $t$.
\end{lemma}
\begin{proof}
Suppose towards contradiction that $\OPT(\mc{I}) > \eta$. As noted earlier, for a fixed $w \in W$, we can think of the variables $\left<w,z\right> \in \textrm{Var}(\mc{I})$ as the truth table of a folded function $f_w:\Z_q^R\mapsto \{0,1\}$, where $\Tf{f_w}(z):= \left<w,z\right>$. This is possible since the variables $\left<w,z\right> \in \textrm{Var}(\mc{I})$ are restricted to $z\in \Z_q^R$ with $z_0 = 0$. Given this alternative point of view, define for every vertex $w\in W$, a set of \emph{candidate labels} $L[w]$ as follows: \begin{align*}
L[w] = \{i\in [R]: \Infl{i}{d}{f_w} \geq \vartheta\}
\end{align*}
Note that $|L[w]|\leq d/\vartheta$ by Observation~\ref{obs:dinfl}.

For every vertex $v\in V$, and every $\mc{W}_v = \{ w_1,\dots,w_t\} \subseteq N(V)$, let \begin{align*}
\mc{C}_{v,\mc{W}_v} := \left\{C_{v,\mc{W}_v,x,S}: x\in \Z_q^R, \,\, S \subseteq[R] \text{ such that } |S| = \epsilon R \right\}
\end{align*}
A standard counting argument then shows that if $\OPT(\mc{I}) > \eta$, then at least $\eta/2$ fraction of the tuples $(v,w_1,\dots,w_t)$ have at least $\eta/2$ fraction of the constraints inside $\mc{C}_{v,\mc{W}_v}$ satisfied. We refer to such tuples as \emph{good}. Adopting the language of folded functions instead of variables, the aforementioned statement can be casted as \begin{align*}
\Pr_{x\in \Z_q^R, S_\epsilon \subseteq [R]} \left[\bigwedge_{i=1}^t \left(f_{w_i}\left(\pi_{v,w_i}\left(C_{x,S_\epsilon}\right)\right) \equiv 1 \right)\right] \geq \eta/2 && \text{if the tuple $(v,w_1,\dots,w_t)$ is good}
\end{align*}
where \begin{align*}
f \left( \pi \left(C_{x,S_\epsilon} \right) \right) \equiv 1 && \Longleftrightarrow && f(z)=1 \quad \quad \forall z \text{ such that } \pi^{-1}(z) \in C_{x,S_\epsilon}
\end{align*}
From Remark~\ref{rem:expec}, we get that $\E\left[f_{w_i} \right] = 1-1/q$, and hence invoking Theorem~\ref{thm:over} on the functions $\bar{f}_{w_1},\dots, \bar{f}_{w_t}$, where $\bar{f}_{w_i}(x):=1-f_{w_i}(\pi_{v,w_i}(x))$, yields that for every good tuple, there exists $\ell_1\neq \ell_2 \in \{1,2,\dots,t\}$ such that $\bar{f}_{w_{\ell_1}}$ and $\bar{f}_{w_{\ell_2}}$ share a common influential coordinate. Note that this is equivalent to saying that there exists $j_1 \in L[w_{\ell_1}]$, $j_2 \in L[w_{\ell_2}]$ such that $\pi_{v,w_{\ell_1}}(j_1) = \pi_{v,w_{\ell_2}}(j_2)$.

We now claim that if $\OPT(\mc{I}) > \eta$, then we can come up with a labeling $\Lambda:V\cup W \mapsto [R]$ that satisfies at least $\frac{\eta \vartheta^2}{2d^2 t^2}$ of edges, which contradicts the fact that $\OPT(\mc{U})\leq \delta$ for a small enough value of $\delta>0$. Towards this end, consider the following randomized labeling procedure: \begin{enumerate}
\item For every $w \in W$, let $\Lambda(w)$ be a random label from the set $L[w]$, or an arbitrary label if $L[w]= \emptyset$.
\item For every $v \in V$, pick a random neighbor $w \in N(v)$ and set $\Lambda(v) = \pi_{v,w}(\Lambda(w))$.
\end{enumerate}
We can readily calculate the fraction of edges in $\mc{U}$ that are satisfied by $\Lambda$. This follows from putting the following observations together: \begin{enumerate}
\item If we pick a random tuple $(v,w_1,\dots,w_t)$, it is \emph{good} with probability $\eta/2$.
\item If $(v,w_1,\dots,w_t)$ is \emph{good}, and we pick $w',w''$ at random from $\{w_1,\dots, w_t\}$, then with probability $1/t^2$ the functions $f_{{w'}}$ and $f_{{w''}}$ share a common influential coordinates.
\item If $(v,w_1,\dots,w_t)$ is \emph{good}, and the functions $f_{{w'}}$ and $f_{{w''}}$ share a common influential coordinates, then picking a random label to ${w'}$ and ${w''}$ from $L[{w'}]$ and $L[{w''}]$ respectively, will satisfies $\pi_{v,{w'}(\Lambda({w'})} = \pi_{v,{w''}(\Lambda({w''})}$ with probability $1/(d^2 / \vartheta^2)$.
\end{enumerate} 
Hence the expected number of edges satisfied by $\Lambda$ in this case is \begin{align*}
\Pr_{vw\in E} \left[ \Lambda(v) = \pi_{v,w}(\Lambda(w)) \right] = \frac{\eta \vartheta^2}{2d^2 t^2}
\end{align*}
\end{proof}
We now show that given an $r$-rounds Sherali-Adams solution of \emph{high} value for $\mc{U}$, we can also come up with an $r$-rounds Sherali-Adams solution for $\mc{I}$ of high value as well. The proof goes along the same lines of that of Lemma~\ref{fcspSA}, and hence we will try to only highlight the differences.
\begin{lemma}
\label{ucspSA}
  Let $\{\mu(S) \mid S \subseteq V \cup W, |S| \leq r\}$ be a consistent collection of local distributions defining a solution to the $r$-rounds Sherali-Adams relaxation of the regular bipartite  \uniquegames instance $\cU$. Then we can define a consistent collection of local distributions $\{\sigma(S) \mid S \subseteq  \textrm{Var}(\cI), |S| \leq r\}$ defining a solution to the $r$-rounds Sherali-Adams relaxation of the $\uCSP$ instance $\cI$ so that
  \begin{align*}
  \E_{C\in \cC}\left[ \Pr_{\alpha \sim \sigma(S_C)}[ \alpha \mbox{ satisfies } C] \right] \geq (1-\epsilon) (1-\frac{1}{q})\left(1  - t \cdot \E_{vw \in E}\left[  \Pr_{(\Lambda(v), \Lambda(w) \sim \mu(\{v,w\})}[\Lambda(v) \neq \pi_{w,v}(\Lambda(w))] \right] \right),
  \end{align*}
  where $t$ and $\epsilon$ are the parameters of the reduction, and $\sigma(S_C)$ is the distribution over the set of variables in the support $S_C$ of constraint $C$.
\end{lemma}

\begin{proof}
Let $\{\mu(S) \mid S \subseteq V \cup W, |S| \leq r\}$ be  a solution to the $r$-rounds SA relaxation of the \uniquegames instance $\cU$, and recall that $\cI$ is the \uCSP instance we get by applying the reduction. We will now use the collection of consistent local distributions of the \uniquegames instance, to construct another  collection of consistent local distributions for the variables in $\textrm{Var}(\cI)$. 

For every set $S\subseteq \textrm{Var}(\cI)$ such that $|S| \leq r$, let $T_S\subseteq W$ be the subset of vertices in the \uniquegames instance defined as follows: \begin{align}
T_S = \{w \in W: \left<w,x\right> \in S\}.
  \end{align}
  We will now construct $\sigma(S)$ from $\mu(T_S)$ in the following manner. 
  Given a labeling $\Lambda_{T_S}$ for the vertices in $T_S$ drawn from $\mu(T_S)$, define an assignment $\alpha_S$ for the variables in $S$ as follows: for a variable $\left< w, x \right> \in S$, let $\ell = \Lambda_{T_S}(w)$ be the label of $w$ according to $\Lambda_{T_S}$. Then the new assignment $\alpha_S$ sets $\alpha_S(\left<w,x\right>) := \Tf{f_\ell}(x)$, where $\Tf{f_\ell}$ is the long code encoding of $\ell$ as in Definition~\ref{def:LC}. The aforementioned procedure defines a family $\{\sigma(S)\}_{S \subseteq \textrm{Var}(\cI): |S| \leq r}$ of local distributions for the variables of the \uCSP instance $\mc{I}$. The same argument as in the proof of Lemma~\ref{fcspSA} yields that $\{\sigma(S) \mid S \subseteq  \textrm{Var}(\cI), |S| \leq r\}$ defines a feasible solution for the $r$-round Sherali-Adams relaxation of the $\uCSP$ instance $\cI$. 
  
  It remains to bound the value of this feasible solution, i.e., \begin{align}
  \label{app:ex1}
    \E_{C\in \cC}\left[ \Pr_{\alpha \sim \sigma(S_C)}[\alpha \mbox{ satisfies } C] \right] =  \E_{v,w_1,\dots,w_t} \left[\Pr_{\Lambda \sim \mu(\{v,w_1,\dots,w_t\}), x ,S} \left[\psi(\Lambda) \mbox{ satisfies } C(v,\mc{W}_v,x,S) \right] \right].
 \end{align}
 where $\psi(.)$ the operator mapping a labeling of the vertices in $T_S$ to an assignment  for the variables in $S$, i.e., $\psi(\Lambda_{T_S}) = \alpha_S$. 
The following claim, which is in some sense the equivalent of Claim~\ref{claim:completeness} in the \uCSP language, along with the same remaining steps of the proof of Lemma~\ref{fcspSA} will yield the proof.
 \begin{claim}
 \label{claim:completeness2}
 If $\Lambda$ satisfies $vw_1,\dots,vw_t$ simultaneously, then $\psi(\Lambda)$ satisfies $C(v,\mc{W}_v,x,S)$ with probability at least $(1-\epsilon)(1-\frac{1}{q})$. Moreover,  if we \emph{additionally} have that  $\Lambda(v) \notin S$ and $x_{\Lambda(v)} \neq 0$, then $\psi(\Lambda)$ always satisfies $C(v,\mc{W}_v,x,S)$.
 \end{claim}

 Equipped with this, we can use conditioning to lower-bound the probability inside the expectation in (\ref{app:ex1}) by a product of two probabilities, where the first is \begin{align}
\label{app:pr1}  \Pr_{\Lambda \sim \mu(\{v,w_1,\dots,w_t\}),x,S} \left[\psi(\Lambda) \mbox{ satisfies } C(v,\mc{W}_v,x,S)  | \Lambda \mbox{ satisfies } vw_1,\dots, vw_t \right]
 \end{align}
 and the second is \begin{align*}
 \Pr_{\Lambda \sim \mu(\{v,w_1,\dots,w_t\})} \left[ \Lambda \mbox{ satisfies } vw_1,\dots, vw_t \right] .
 \end{align*}
 Thus using Claim \ref{claim:completeness2}, we get 
\begin{align*}
 & \E_{C\in \cC}\left[ \Pr_{\alpha \sim \sigma(S_C)}[\alpha \mbox{ satisfies } C] \right] \geq  (1-\epsilon)(1-\frac{1}{q})\cdot \E_{v,w_1,\dots,w_t} \left[ \Pr_{\Lambda \sim \mu(\{v,w_1,\dots,w_t\})} \left[ \Lambda \mbox{ satisfies } vw_1,\dots, vw_t \right]  \right] \\
\geq & (1-\epsilon) (1-\frac{1}{q}) \cdot\left( 1 -  t \cdot \E_{v,w} \left[\Pr_{\Lambda \sim \mu(\{v,w\})} \left[ \Lambda \mbox{ does not satisfy } vw \right] \right]\right)
 \end{align*}
\end{proof}

The proof of Corollary~\ref{cor:fcspSA} adjusted to the \uCSP problem now yields Theorem~\ref{thm:mainUCSP}.

%%%%%%%%%%%%%%%%% LP Reduction %%%%%%%%%%%%%%%%%

\subsection{LP-reduction from \uCSP to \ekvertexcover}
We will now reduce \uCSP to \ekvertexcover on $q$-Uniform hypergraphs with the reduction mechanism
outlined in Section~\ref{sec:lpreduc}, which will yield the desired LP
hardness for the latter problem. 

We start by recasting \ekvertexcover and \uCSP
in the language of Section~\ref{sec:lpreduc}. The first problem is defined on a fixed
$q$-uniform hypergraph $H = (V,E)$.

\begin{problem}[$\ekvertexcover(G)$]
\label{prob:ekvc}
  The set of feasible solutions $\mc{S}$ consists of all possible 
  vertex covers $U \subseteq V$, and there is one instance 
  $\mc{I} = \mc{I}(H') \in \mf{I}$ for each induced subgraph $H'$ of
  $G$. For each vertex cover \(U\) we have \(\Cost_{\mc{I}(H')}(U) \coloneqq
  |U \cap V(H')|\) being the size of the induced vertex cover in \(H'\). 
\end{problem}
  
We also recast \uCSP as follows. Let $n, q, k \in \N$ be fixed,
with $k \leq n$.

\begin{problem}[$\uCSP(n,q,k)$]
 \label{prob:ucsp}
 The set of feasible solutions \(\mc{S}\) consists of all 
 possible variable assignments, i.e., all possible values of $\Z_q^n$
 and there is one instance 
 $\mc{I} = \mc{I}(\mc{P})$ for each possible set $\mc{P} =
 \{P_1, \ldots, P_m\}$ of
\uCSP predicates of arity $k$. As before, for 
 an instance \(\mc{I} \in \mf{I}\) and an assignment 
 \(x \in \Z_q^n\), \(\Val_{\mc{I}}(x)\) is the fraction 
 of predicates $P_i$ that $x$ satisfies (see Definition~\ref{def:ucsp}). 
\end{problem}

With the notion of LP relaxations and \uCSP from above, we can now formulate LP-hardness of approximation for \uCSP{}s, which   follows directly from Theorem~\ref{thm:mainUCSP} by
the result of \cite{CLRS13} (See the discussion in \cite{CLRS13} and Section 7 in \cite{LRS14}).

\begin{theorem}
\label{mainucsp}
For every $\epsilon>0$ and  alphabet size $q\geq2$, there exists a constant arity $k = k(\epsilon)$ 
such that for infinitely many $n$ we have \(\fc_+(\uCSP(n,q,k),1-1/q - \epsilon,\epsilon) \geq
n^{\Omega\left(\log n / \log \log n \right)}\).
\end{theorem}

Similar to Section~\ref{sec:lphardness}, we first define our host hypergraph, and then provide a reduction that will yield our hardness result for \ekvertexcover using Theorem~\ref{BPZmain}.

\begin{definition}[\ekvertexcover host hypergraph]
  For fixed number of variables $n$, alphabet $q$, and arity $k \leq n$ 
  we define a hypergraph \(H^* = H^*(n,q,k)\) as follows. Let
  $x_1$, \ldots, $x_n$ denote the variables of the CSP.

  \emph{Vertices:} For every subset $S =\{i_1,\dots,i_k\} \subseteq [n]$, and every value of $A = (a_1,\dots,a_k) \in \Z_q^k$,
we have a vertex $v_{S,A}$ corresponding to the \uCSP predicate\begin{align*}
P(x_{i_1,},\dots,x_{i_k}) = 1 \quad \quad \text{ if and only if } \quad \quad \bigwedge_{j=1}^k (x_{i_j} \neq a_j)
\end{align*} 

%\emph{Hyperedges:} Any $q$ vertices $v_{S_1,A_1},\dots,v_{S_q,A_q}$ are connected with a hyperedge if there exists $i \in \bigcap_{j=1}^q S_j$ such that $a^1_i \neq a^2_i \neq \dots \neq a^q_i$, where, for $\ell \in \{1,2,\dots,q\}$, $a_i^\ell$  is the $i^{th}$ entry of the vector $A_\ell$ corresponding to the variable $x_i$ in the predicate defined by the pair $(S_\ell,A_\ell)$. In other words, we have a hyperedge connecting $q$ vertices sharing a \emph{common variable} $x_i$, if no two of their \emph{corresponding} predicates check $x_i$ versus the same $a \in \Z_q$.  

\emph{Hyperedges:} Any $q$ vertices $v_{S_1,A_1},\dots,v_{S_q,A_q}$ are connected with a hyperedge if there exists a variable $x_i \in \bigcap_{j=1}^q S_j$ 
%indexed by $i_1,\dots,i_q $ in the vectors $A_1,\dots,A_q$ respectively
, such that $a_{i_1} \neq a_{i_2} \neq \dots \neq a_{i_q}$, where $a_{i_j}$ is the entry of the vector $A_j$ that is compared versus the variable $x_i$ in the predicate defined by the pair $(S_j,A_j)$. %, where, for $\ell \in \{1,2,\dots,q\}$, $a_i^\ell$  is the $i^{th}$ entry of the vector $A_\ell$ corresponding to the variable $x_i$ in the predicate defined by the pair $(S_\ell,A_\ell)$. 
In other words, we have a hyperedge connecting $q$ vertices sharing a \emph{common variable} $x_i$, if no two of their \emph{corresponding} predicates check $x_i$ versus the same $a \in \Z_q$.  
\end{definition}

Note that the graph has \(q^k \binom{n}{k}\) vertices, which
is polynomial in \(n\) for fixed \(k\) and $q$. In order to establish LP-inapproximability 
of \ekvertexcover it now suffices to define a reduction 
satisfying Theorem~\ref{BPZmain}.

\begin{maintheorem}
  For every $\epsilon>0, q \geq 2$ and for infinitely many $n$, there exists a 
  hypergraph $H$ with $|V(H)|=n$ such that
  $\fc_+(\ekvertexcover(H),q-\epsilon) \geq n^{\Omega \left( \log n / \log \log n\right)}$. \end{maintheorem}
\begin{proof}
We reduce \uCSP on $n$ variables of alphabet $\Z_q$ with sufficiently large arity $k = k(\epsilon)$
to \ekvertexcover over $H \coloneqq H^*(n,q,k)$. 
For a \uCSP instance $\mc{I}_1 = \mc{I}_1(\mc{P})$ and set of Not-Equal
predicates $\mc{P} = \{P_{S_1,A_1}, P_{S_2,A_2},\ldots,P_{S_m,A_m}\}$, let $H(\mc{P})$ be the induced 
subgraph of $G$ on the set of vertices
$V(\mc{P}) = \{v_{S_i,A_i} \mid 1\leq i \leq m\}$. 

Similarly to Section~\ref{sec:lphardness}, we provide maps defining a reduction from \uCSP 
to \ekvertexcover.  The proof will then follow by combining Theorems~\ref{mainucsp} and~\ref{BPZmain}.

In the following, let \(\Pi_1 = (\mc{S}_1, \mf{I}_1)\) be the 
\uCSP problem and let \(\Pi_2 = (\mc{S}_2, \mf{I}_2)\) be the 
\ekvertexcover problem. In view of Definition~\ref{RedDef}, we 
map $\mc{I}_1 = \mc{I}_1(\mc{P})$ to $\mc{I}_2 = \mc{I}_2(H(\mc{P}))$
and let \(\mu \coloneqq 1\) and \(\zeta_{\mc{I}_1} \coloneqq \frac{1}{m}\) 
where \(m\) is the number of constraints in \(\mc{P}\).

For a total assignment $x \in \mc{S}_1$ we define
$U = U(x) \coloneqq \{v_{S,A}: P_{S,A}(x) = 0\}$. The latter
is indeed a vertex cover. To see this, consider its complement $I = I(x) \coloneqq
\{ v_{S,A} \mid P_{S,A}(x) = 1\}$. Since $x$ satisfies all the constraints corresponding to vertices in $I$ simultaneously, no hyperedge can be completely contained in $I$. Otherwise this would imply that there exists a variable $x_i$, and $q$ predicates $P^{'}_1, P^{'}_2,\dots, P^{'}_q \in \mc{P}$ requiring $x_i \neq j$ for all $j\in \Z_q$, and yet are all simultaneously satisfied by $x$.

We first verify the condition that \(\Val_{\mc{I}_1}(x) = 1 - \frac{1}{m}
\Cost_{\mc{I}_2}(U(x))\) for all instances \(\mc{I}_1 \in \mf{I}_1\) and 
assignments \(x \in \mc{S}_1\). Every predicate $P_{S,A}$ in $\mc{P}$ over 
the variables in $\{x_{i} \mid i \in S\}$ has exactly one 
representative vertex $v_{S,A}$, that will be inside $U$ only if $P_{S,A}(x)=0$, and hence our claim holds.
In other words, for any specific \(\mc{P}\) the affine shift is \(1\),
and the normalization factor is $\frac{1}{m}$. 

Next we verify exactness of the reduction, i.e.,
\[\OPT(\mc{I}_1) = 1 - \frac{1}{m} \OPT(\mc{I}_2)\,.\]
For this take an arbitrary vertex cover $U \in \mc{S}_2$ 
of $H$ and consider its complement. This is an independent 
set, say $I$. 
As $I$ is an independent set\footnote{In a hypergraph $H=(V,E)$ a set $I \subseteq V$ is said to be \emph{independent} if no hyperedge of $H$ is fully contained in $I$.}, we know 	that for any variable $x_\ell$ in $ \bigcup_{v_{S,A} \in I} S$, there exist a least one $\tilde{a}_{x_\ell}\in \Z_q$ such that  $x_\ell$ is \emph{not checked} versus $\tilde{a}_{x_\ell}$ in any of the predicates corresponding to vertices in $I$. Hence any assignment $x$ setting each $x_\ell$ to $\tilde{a}_{x_\ell}$ as defined earlier, sets $P_{S,A}(x) = 1$ for all $v_{S,A} \in I$.
%As \(I\) is an independent set, then there exists an assignment for $x$ that sets $P_{S,A}(x) =  1$ for all $v_{S,A}\in I$. 
Then the corresponding vertex cover $U(x)$
is contained in $U$. Thus there always exists an optimum solution 
to $\mc{I}_2$ that is of the form $U(x)$. Therefore, the reduction 
is exact.

It remains to compute the inapproximability factor via
Theorem~\ref{BPZmain}. 
We have 
\begin{align*}
  \rho_2 =  \frac{1 -
  \epsilon}{1 - (1 - 1/q - \epsilon)}
\geq q - \Theta (\epsilon)
\end{align*}
\end{proof}

%%% Local Variables:
%%% mode: latex
%%% TeX-master: "ExtendedFormUGCtoVC"
%%% End:

%  LocalWords:  ekvertexcover qvertexcover newcommand vertexcover neq
%  LocalWords:  sconstraintsatisfaction independentset sec-1fcsp emph
%  LocalWords:  uniquegames vertexcover vertexcover ucsp arity ldots
%  LocalWords:  rightarrow qquad mbox qquad bigwedge geq subseteq thm
%  LocalWords:  Sherali-Adams varepsilon leq Sherali mapsto q-ary th
%  LocalWords:  subsubsection Infl mathbb dinfl forall vartheta equiv
%  LocalWords:  oplus expec qsound langle textrm pageref mdframed vw
%  LocalWords:  linewidth emptyset fcspSA cdot hypergraphs lpreduc fc
%  LocalWords:  hypergraph ekvc coloneqq lphardness BPZmain hyperedge
%  LocalWords:  Hyperedges bigcap LP-inapproximability
%  LocalWords:  inapproximability ExtendedFormUGCtoVC

\section{SDP-Hardness for \independentset}
\label{sec:sdpind}
We saw in Section~\ref{sec:lphardness} how to obtain an LP-hardness for \vertexcover and \independentset, starting from an LP-hardness for the \oFk problem. Restricting our starting CSP to have only \emph{one free bit} is crucial for the \vertexcover problem, since each constraint is then represented by a \emph{cloud} containing exactly two vertices in the resulting graph. In this case, an assignment satisfying \emph{almost all} the constraints, corresponds to a vertex cover containing  \emph{slightly more than half} of the vertices (i.e., one vertex in almost all the clouds, and both vertices in the \emph{unsatisfied} clouds), whereas if no assignment can simultaneously satisfy more than $\epsilon$-fraction of the constraints, then any vertex cover should contain \emph{almost all} the vertices. This extreme behaviour of the resulting graph is necessary to obtain a gap of $2$ for the \vertexcover problem.

However, if we are only interested in the \independentset problem, any CSP with a sufficiently large gap between the soundness and completeness can yield the desired LP-Hardness, by virtue of the well-known FGLSS reduction \cite{FGLSS96}. Formally speaking, given reals $0<s<c \leq 1$, and any CSP problem $\Pi(P,n,k)$, where $n$ is the number of variables and $P$ is a predicate of arity $k$,  and knowing that no small linear program can provide a $(c,s)$-approximation for this CSP, then one can show that no small LP can as well approximate the \independentset problem within a factor of $c/s$. This can be simply done by tweaking the reduction of Section~\ref{sec:lphardness} in a way that the number of vertices in each cloud is equal to the number of satisfying assignments for the predicate. Hence dropping the \emph{one free bit} requirement, and restricting ourselves to CSPs such that $c/s = 1/\epsilon$ for arbitrarily small $\epsilon:= \epsilon(k)>0$, would yield the desired $\omega(1)$ LP-hardness for the \independentset problem. 

Moreover, the reduction framework of~\cite{BPZ} and our construction in Section~\ref{sec:lphardness} are agnostic to whether we are proving LP or SDP lower bounds, and hence having an analog of Theorem~\ref{mainfcsp} in the SDP world would yield that any SDP of size less than $n^{\Omega \left(\log n / \log \log n\right)}$ has an integrality gap of $\omega(1)$ for the \independentset problem. In fact such SDP-hardness results for certain families of CSPs and hence an analog of Theorem~\ref{mainfcsp} are known: if our starting CSP has a predicate that supports pairwise independence with a sufficiently large arity $k$, then the result of~\cite{BCK15} by virtue of \cite{LRS14} gives us the desired SDP base hardness. By the argumentation from above we obtain:

\begin{corollary}
For every $\epsilon>0$ and for infinitely many $n$, there exists a graph $G$ with $|V(G)| = n$, such that no polynomial size SDP is a $(1/\epsilon)$-approximate SDP relaxation for $\independentset(G)$.
\end{corollary}
%%% Local Variables:
%%% mode: latex
%%% TeX-master: "ExtendedFormUGCtoVC"
%%% End:

\section{Discussion of related problems}

We believe that our approach extends to many other related problems.
As proved here, it applies to \ekvertexcover. %We think that it will
%extend to other problems beyond this.
Moreover, we would like to stress that our reduction is agnostic 
to whether it is used for LPs or SDPs and Lasserre gap instances 
for \oFk, together with \cite{LeeRS14} and our reduction would
provide SDP hardness of approximation for \vertexcover. 
This already holds for the \independentset problem   as we saw in 
Section~\ref{sec:sdpind}, since the starting CSP does not need to have 
only one free bit, as long as the gap between the soundness and 
completeness is sufficiently large.

Note that we are only able to establish hardness of approximations for
the stable set problem within any constant factor, while assuming \(P
\neq NP\) one can establish hardness of approximation within
\(n^{1-\epsilon}\). The reason for this gap is that the standard
amplification techniques via graph products do not fall into the
reduction framework in \cite{BPZ}. Also, there will be limits to
amplification as established by the upper bounds in
Section~\ref{sec:upperbounds}. 

Finally, we would like to remark
that our lower bounds on the size can be probably further
strengthened, however, with our current reductions this would require a strengthened version of the results in \cite{CLRS13}.

\section*{Acknowledgements}
Research reported in this paper was partially supported by NSF CAREER
award CMMI-1452463, NSF grant CMMI-1300144, ERC Starting Grant 
335288-OptApprox, and ERC Consolidator Grant 615640-ForEFront. Research was partially conducted at the Oberwolfach Workshop 1446 and Dagstuhl Workshop 15082.

\bibliographystyle{abbrv}
\bibliography{references}

\begin{thebibliography}{10}

\bibitem{BCVGZ}
M.~C. Aditya~Bhaskara, A.~Vijayaraghavan, V.~Guruswami, and Y.~Zhou.
\newblock Polynomial integrality gaps for strong {SDP} relaxations of {D}ensest
  $k$-subgraph.
\newblock In {\em Proc.\ SODA 2012}, pages 388--405, 2012.

\bibitem{ABL02}
S.~Arora, B.~Bollob{\'a}s, and L.~Lov{\'a}sz.
\newblock Proving integrality gaps without knowing the linear program.
\newblock In {\em Proc.\ FOCS 2002}, pages 313--322, 2002.

\bibitem{ABLT06}
S.~Arora, B.~Bollob{\'a}s, L.~Lov{\'a}sz, and I.~Tourlakis.
\newblock Proving integrality gaps without knowing the linear program.
\newblock {\em Theory Comput.}, 2:19--51, 2006.

\bibitem{BK}
N.~Bansal and S.~Khot.
\newblock Optimal long code test with one free bit.
\newblock In {\em Proc.\ FOCS 2009}, FOCS '09, pages 453--462, Washington, DC,
  USA, 2009. IEEE Computer Society.

\bibitem{BCK15}
B.~Barak, S.~O. Chan, and P.~K. Kothari.
\newblock Sum of squares lower bounds from pairwise independence.
\newblock In {\em Proceedings of the Forty-Seventh Annual ACM on Symposium on
  Theory of Computing}, STOC '15, pages 97--106, New York, NY, USA, 2015. ACM.

\bibitem{BGMT12}
S.~Benabbas, K.~Georgiou, A.~Magen, and M.~Tulsiani.
\newblock {SDP} gaps from pairwise independence.
\newblock {\em Theory Comput.}, 8(12):269--289, 2012.

\bibitem{Bienstock08}
D.~Bienstock.
\newblock Approximate formulations for $0$-$1$ knapsack sets.
\newblock {\em Operations Research Letters}, 36:317--320, 2008.

\bibitem{BFPS12}
G.~Braun, S.~Fiorini, S.~Pokutta, and D.~Steurer.
\newblock {Approximation Limits of Linear Programs (Beyond Hierarchies)}.
\newblock In {\em Proc.\ FOCS 2012}, pages 480--489, 2012.

\bibitem{BP13}
G.~Braun and S.~Pokutta.
\newblock {Common information and unique disjointness}.
\newblock In {\em Proc.\ FOCS 2013}, pages 688--697, 2013.
\newblock \url{http://eccc.hpi-web.de/report/2013/056/}.

\bibitem{BPZ}
G.~Braun, S.~Pokutta, and D.~Zink.
\newblock Inapproximability of combinatorial problems via small {LP}s and
  {SDP}s.
\newblock In {\em Proceedings of STOC 2015}, pages 107--116, New York, NY, USA,
  2015. ACM.

\bibitem{BM13}
M.~Braverman and A.~Moitra.
\newblock An information complexity approach to extended formulations.
\newblock In {\em Proc.\ STOC 2013}, pages 161--170, 2013.

\bibitem{BDP13}
J.~Bri\"et, D.~Dadush, and S.~Pokutta.
\newblock On the existence of 0/1 polytopes with high semidefinite extension
  complexity.
\newblock In {\em Proc.\ ESA 2013}, pages 217--228, 2013.

\bibitem{CLRS13}
S.~O. Chan, J.~R. Lee, P.~Raghavendra, and D.~Steurer.
\newblock {Approximate Constraint Satisfaction Requires Large LP Relaxations}.
\newblock {\em Proc.\ FOCS 2013}, 0:350--359, 2013.

\bibitem{CMM}
M.~Charikar, K.~Makarychev, and Y.~Makarychev.
\newblock {Integrality Gaps for Sherali-Adams Relaxations}.
\newblock In {\em Proc.\ STOC 2009}, STOC '09, pages 283--292, New York, NY,
  USA, 2009. ACM.

\bibitem{CCZ10}
M.~Conforti, G.~Cornu{\'e}jols, and G.~Zambelli.
\newblock Extended formulations in combinatorial optimization.
\newblock {\em 4OR}, 8:1--48, 2010.

\bibitem{FM07b}
W.~F. de~la Vega and C.~Kenyon-Mathieu.
\newblock Linear programming relaxations of {M}axcut.
\newblock In {\em Proc.\ SODA 2007}, pages 53--61, 2007.

\bibitem{DS02}
I.~Dinur and S.~Safra.
\newblock The importance of being biased.
\newblock In {\em Proc.\ STOC 2002}, pages 33--42, 2002.

\bibitem{DS05}
I.~Dinur and S.~Safra.
\newblock On the hardness of approximating minimum vertex cover.
\newblock {\em Annals of Mathematics}, 162(1):439--485, 2005.

\bibitem{feige1991approximating}
U.~Feige, S.~Goldwasser, L.~Lov{\'a}sz, S.~Safra, and M.~Szegedy.
\newblock {Approximating clique is almost NP-complete}.
\newblock In {\em Proc.\ FOCS 1991}, pages 2--12. IEEE Comput. Soc. Press,
  1991.

\bibitem{FGLSS96}
U.~Feige, S.~Goldwasser, L.~Lov\'asz, S.~Safra, and M.~Szegedy.
\newblock Interactive proofs and the hardness of approximating cliques.
\newblock {\em J. ACM}, 43:268--292, 1996.

\bibitem{FJ14}
U.~Feige and S.~Jozeph.
\newblock Demand queries with preprocessing.
\newblock In {\em Proc.\ ICALP 2014}, pages 477--488, 2014.

\bibitem{FMPTW12}
S.~Fiorini, S.~Massar, S.~Pokutta, H.~R. Tiwary, and R.~{de}~Wolf.
\newblock {Linear vs. Semidefinite Extended Formulations: Exponential
  Separation and Strong Lower Bounds}.
\newblock {\em Proc. STOC 2012}, pages 95--106, 2012.

\bibitem{GM08}
K.~Georgiou and A.~Magen.
\newblock Limitations of the sherali-adams lift and project system:
  Compromising local and global arguments.
\newblock {\em Technical Report CSRG-587}, 2008.

\bibitem{GMPT07}
K.~Georgiou, A.~Magen, T.~Pitassi, and I.~Tourlakis.
\newblock Integrality gaps of {$2-o(1)$} for vertex cover {SDP}s in the
  {L}ov\'asz-{S}chrijver hierarchy.
\newblock In {\em Proc.\ FOCS 2007}, pages 702--712, 2007.

\bibitem{GeorgiouMT09}
K.~Georgiou, A.~Magen, and M.~Tulsiani.
\newblock {Optimal Sherali-Adams Gaps from Pairwise Independence}.
\newblock In {\em In Proc.\ APPROX 2009}, pages 125--139, 2009.

\bibitem{haastad1996clique}
J.~H{\aa}stad.
\newblock Clique is hard to approximate within n 1-\&epsiv.
\newblock In {\em Foundations of Computer Science, 1996. Proceedings., 37th
  Annual Symposium on}, pages 627--636. IEEE, 1996.

\bibitem{Hochbaum97}
D.~Hochbaum.
\newblock Approximating covering and packing problems: Set cover, vertex cover,
  independent set and related problems.
\newblock In {\em Approximation Algorithms for NP-hard Problems}. PWS
  Publishing Company, 1997.

\bibitem{Kaibel11}
V.~Kaibel.
\newblock Extended formulations in combinatorial optimization.
\newblock {\em Optima}, 85:2--7, 2011.

\bibitem{KPT10}
V.~Kaibel, K.~Pashkovich, and D.~Theis.
\newblock Symmetry matters for the sizes of extended formulations.
\newblock In {\em Proc.\ IPCO 2010}, pages 135--148, 2010.

\bibitem{KMN11}
A.~R. Karlin, C.~Mathieu, and C.~T. Nguyen.
\newblock {Integrality Gaps of Linear and Semi-definite Programming Relaxations
  for Knapsack}.
\newblock In {\em Proc.\ IPCO 2011}, pages 301--314, 2011.

\bibitem{Khot}
S.~Khot.
\newblock On the power of unique 2-prover 1-round games.
\newblock In {\em In Proc.\ STOC 2002}, STOC '02, pages 767--775, New York, NY,
  USA, 2002. ACM.

\bibitem{KR03}
S.~Khot and O.~Regev.
\newblock Vertex cover might be hard to approximate to within $2 -
  \varepsilon$.
\newblock In {\em In Proc.\ CCC 2003}, pages 379--386, 2003.

\bibitem{KR08}
S.~Khot and O.~Regev.
\newblock Vertex cover might be hard to approximate to within $2 -
  \varepsilon$.
\newblock {\em J. Comput. System Sci.}, pages 335--349, 2008.

\bibitem{KS09}
S.~Khot and R.~Saket.
\newblock {SDP} {I}ntegrality {G}aps with {L}ocal $\ell_1$-{E}mbeddability.
\newblock In {\em Proc.\ FOCS 2009}, pages 565--574, 2009.

\bibitem{Lasserre01a}
J.~B. Lasserre.
\newblock An explicit exact {SDP} relaxation for nonlinear $0$-$1$ programs.
\newblock In {\em Proc.\ IPCO 2001}, pages 293--203, 2001.

\bibitem{Lasserre01b}
J.~B. Lasserre.
\newblock Global optimization with polynomials and the problem of moments.
\newblock {\em SIAM Journal on Optimization}, pages 796--817, 2001.

\bibitem{LRS14}
J.~Lee, P.~Raghavendra, and D.~Steurer.
\newblock Lower bounds on the size of semidefinite programming relaxations.
\newblock arXiv:1411.6317, 2014.

\bibitem{LeeRS14}
J.~R. Lee, P.~Raghavendra, and D.~Steurer.
\newblock Lower bounds on the size of semidefinite programming relaxations.
\newblock {\em CoRR}, abs/1411.6317, 2014.

\bibitem{LS91}
L.~Lov{\'a}sz and A.~Schrijver.
\newblock Cones of matrices and set-functions and 0-1 optimization.
\newblock {\em SIAM Journal on Optimization}, 1(2):166--190, 1991.

\bibitem{MOO05}
E.~Mossel, R.~O'Donnell, and K.~Oleszkiewicz.
\newblock Noise stability of functions with low influences: invariance and
  optimality.
\newblock In {\em Foundations of Computer Science, 2005. FOCS 2005. 46th Annual
  IEEE Symposium on}, pages 21--30. IEEE, 2005.

\bibitem{RyanBook}
R.~O'Donnell.
\newblock {\em Analysis of Boolean Functions}.
\newblock Cambridge University Press, New York, NY, USA, 2014.

\bibitem{Parillo00}
P.~Parrilo.
\newblock {\em Structured semidefinite programs and semialgebraic geom- etry
  methods in robustness and optimization}.
\newblock PhD thesis, California Institute of Technology, 2000.

\bibitem{Pashkovich12}
K.~Pashkovich.
\newblock {\em Extended Formulations for Combinatorial Polytopes}.
\newblock PhD thesis, Magdeburg Universit\"at, 2012.

\bibitem{RS09}
P.~Raghavendra and D.~Steurer.
\newblock {I}ntegrality {G}aps for {S}trong {SDP} {R}elaxations of {U}nique
  games.
\newblock In {\em Proc.\ FOCS 2009}, pages 575--585, 2009.

\bibitem{Rothvoss11}
T.~Rothvo{\ss}.
\newblock Some 0/1 polytopes need exponential size extended formulations, 2011.
\newblock arXiv:1105.0036.

\bibitem{Rothvoss13}
T.~Rothvo{\ss}.
\newblock The matching polytope has exponential extension complexity.
\newblock {\em Proceedings of STOC}, 2014.

\bibitem{Schoenebeck08}
G.~Schoenebeck.
\newblock Linear level lasserre lower bounds for certain k-{CSP}s.
\newblock In {\em Proc.\ FOCS 2008}, pages 593--602, 2008.

\bibitem{STT07}
G.~Schoenebeck, L.~Trevisan, and M.~Tulsiani.
\newblock Tight integrality gaps for lov\'asz-schrijver lp relaxations of
  vertex cover and max cut.
\newblock In {\em In Proc.\ STOC 2007~}, pages 302--310, 2007.

\bibitem{SA90}
H.~Sherali and W.~Adams.
\newblock A hierarchy of relaxations between the continuous and convex hull
  representations for zero-one programming problems.
\newblock {\em SIAM Journal on Discrete Mathematics}, 3:411--430, 1990.

\bibitem{sherali1990hierarchy}
H.~D. Sherali and W.~P. Adams.
\newblock A hierarchy of relaxations between the continuous and convex hull
  representations for zero-one programming problems.
\newblock {\em SIAM Journal on Discrete Mathematics}, 3(3):411--430, 1990.

\bibitem{Singh10}
M.~Singh.
\newblock Bellairs workshop on approximation algorithms.
\newblock Open problem session \#1, 2010.

\bibitem{O12}
O.~Svensson.
\newblock Hardness of vertex deletion and project scheduling.
\newblock In {\em Approximation, Randomization, and Combinatorial Optimization.
  Algorithms and Techniques}, pages 301--312. Springer, 2012.

\bibitem{Tulsiani09}
M.~Tulsiani.
\newblock {CSP} gaps and reductions in the {L}asserre hierarchy.
\newblock In {\em Proc.\ STOC 2009}, pages 303--312, 2009.

\bibitem{Yannakakis88}
M.~Yannakakis.
\newblock Expressing combinatorial optimization problems by linear programs
  (extended abstract).
\newblock In {\em Proc.\ STOC 1988}, pages 223--228, 1988.

\bibitem{Yannakakis91}
M.~Yannakakis.
\newblock Expressing combinatorial optimization problems by linear programs.
\newblock {\em J. Comput. System Sci.}, 43(3):441--466, 1991.

\end{thebibliography}

\newpage

\appendix

\section{Definition of Sherali-Adams  for General Binary Linear Programs}
\label{sec:genSA}
For completeness, we give the general definition of the $r$-rounds SA tightening
of a given LP, and then we show that for CSPs the obtained relaxation is
equivalent to~\eqref{eq:CSP_prog2}.

Consider the following Binary Linear Program for $c \in 
\mathbb{R}^{n}$, $A \in \mathbb{R}^{m \times n}$ 
and $b \in \mathbb{R}^{m \times 1}$:
\begin{equation*}
\begin{array}{rll}
\max &\displaystyle \sum_{i=1}^n c_i x_i\\[3ex]
\text{s.t.}&Ax \leqslant b\\
&x \in \{0,1\}^n\,.
\end{array}
\end{equation*}
By replacing the integrality constraint with
\(0 \leq x \leq 1\), we get an LP relaxation. 

Sherali and Adams \cite{sherali1990hierarchy} proposed a systematic 
way for tightening such relaxations, by reformulating them in a higher
dimensional space. Formally speaking, the $r$-rounds SA relaxation is
obtained by multiplying each base inequality $\sum_{j=1}^n A_{ij} x_j
\leq b_j$ by $\prod_{s\in S} x_s \prod_{t\in T}(1-x_t)$ for all disjoint 
$S, T \subseteq [n]$ such that $|S \cup T| < r$. This gives the following
set of polynomial inequalities for each such pair  $S$ and $T$:
\begin{align*}
\left(\sum_{j\in [n]} A_{ij} x_j\right) \prod_{s\in S} x_s \prod_{t\in T} (1-x_t)  \leq & b_i \prod_{s\in S} x_s \prod_{t\in T} (1-x_t)  
\quad &\forall i\in[m]\,, \\[3ex]
\displaystyle 0 \leq  x_j \prod_{s \in S} x_s &\prod_{t\in T} (1-x_t) \leq  1 & \forall j\in [n]\, . 
\end{align*}
These constraints are then linearized by first expanding (using $x_i^2 = x_i$,
and thus $x_i(1-x_i) = 0$), and then replacing each monomial $\prod_{i \in H}
x_i$ by a new variable $y_H$, where $H \subseteq [n]$ is a set of size at most
$r$. Naturally, we set $y_\emptyset := 1$. This gives us the following linear program, referred to as the $r$-rounds SA relaxation:
\begin{equation*}
\begin{array}{rll}
\max &\displaystyle \sum_{i=1}^n c_i y_{\{i\}}\\[3ex]
\text{s.t.}&\displaystyle 
\sum_{H \subseteq T} (-1)^{|H|}\left(\sum_{j\in [n]} A_{ij} y_{H \cup S \cup \{j\}} \right) \leq b_i \sum_{H \subseteq T} (-1)^{|H|}y_{H \cup S}  
\quad &\forall i\in[m], S, T\,,\\[5ex]
& \qquad \quad \displaystyle 0 \leq  \sum_{H \subseteq T} (-1)^{|H|}y_{H\cup S \cup \{j\}} \leq 1 
\quad &\forall j \in [n], \forall S, T,\\[3ex]
& \qquad \qquad \qquad \quad \qquad \qquad y_\emptyset = 1
\end{array}
\end{equation*}
where in the first two constraint we take $S, T \subseteq [n]$ 
with $S \cap T = \emptyset$ and $|S \cup T| < r$. 

One could go back to the original space by letting $x_i = y_{\{i\}}$
and projecting onto the \(x\),
however we will refrain from doing that, in order to be able to write
objective functions that are not linear but degree-$k$ polynomials, as
is natural in the context of CSPs of arity $k$. Since we need to do $k$ rounds of
SA before even being able to write the objective function as a linear
function, it makes more sense to work in higher dimensional space.

For \constraintsatisfaction, the canonical $r$-rounds
SA relaxation is defined as follows. 
Consider any CSP defined over $n$ 
variables $x_1, \ldots, x_n \in [R]$, with $m$ constraints 
$\mc{C} = \{C_1,\dots,C_m\}$ where the arity of each constraint 
is at most $k$. For each $j \in [n]$ and $u \in [R]$, we introduce 
a binary variable $x(j,u)$, meant to be the indicator of $x_j = u$. 
Using these variables, the set of feasible assignments can naturally be formulated as 
\begin{align*}
  \sum_{u\in [R]} x(j,u) = 1 & \qquad \forall j \in [n]\,,\\
  x(j,u) \in \{0,1\} &\qquad \forall j \in [n], u \in [R]\,.
\end{align*}
If we relax the integrality constraints by, for each $j\in [n]$, $u\in [R]$, replacing $x(j,u) \in \{0,1\}$ by $x(i,u) \geq 0$ (we omit the upper bounds 
of the form $x(j,u) \leq 1$ as they are already implied by the other 
constraints) then we obtain the following constraints for the $r$-rounds SA relaxation : 
\begin{align*}
  \sum_{H \subseteq  T} (-1)^{|H|}\sum_{u\in [R]}  y_{H \cup S \cup \{(j,u)\}} & = \sum_{ H \subseteq  T} (-1)^{|H|}y_{H \cup S}  
&\forall j\in[n], S, T\,,\\[5ex]
\sum_{ H \subseteq  T} (-1)^{|H|}&  y_{H \cup S \cup \{(j,u)\}}  \geq 0  &\forall (j,u)\in[n]\times [R], S, T\,,
\end{align*}
where  we take $S, T \subseteq [n]\times[R]$ 
with $S \cap T = \emptyset$ and $|S \cup T| < r$. 

To simplify the above description, we observe that we only need the constraints for which $T=\emptyset$.
\begin{claim}
  All the above constraints are implied by the subset of constraints for which $T=\emptyset$.
\end{claim}
\begin{proof}
  The equality constraints are easy to verify since $\sum_{u\in [R]} y_{S \cup \{(j,u)\}} = y_S$ for all $S\subseteq [n]\times [R]$ with $|S| < r$  implies
  \begin{align*}
    \sum_{S \subseteq H \subseteq S \cup T} (-1)^{|H \cap T|}\sum_{u\in [R]}  y_{H \cup \{(j,u)\}} & = \sum_{S \subseteq H \subseteq S \cup T} (-1)^{|S \cap T|}y_H.  
  \end{align*}

  Now consider the inequalities. If we let $T = \{(j_1, u_1), (j_2, u_2), \ldots, (j_\ell, u_\ell)\}$ then by the above equalities
  \begin{align*}
    \sum_{ H \subseteq  T} (-1)^{|H|}  y_{H \cup S \cup \{(j,u)\}}  & = \sum_{ H \subseteq  T\setminus \{(j_1,u_1)\}} (-1)^{|H|}  y_{H \cup S \cup \{(j,u)\}} - \sum_{ H \subseteq  T\setminus \{(j_1, u_1)\}} (-1)^{|H|}  y_{H  \cup S \cup \{(j,u),(j_1, u_1)\}} \\
    & = \sum_{u'_1 \in [R]: u'_1 \neq u_1} \sum_{ H \subseteq  T\setminus
      \{(j_1,u_1)\}} (-1)^{|H|}  y_{H \cup S \cup \{(j,u),(j_1, u'_1)\}}\\
    &~~\vdots \\
    & = \sum_{u'_t \in [R]: u'_t \neq u_t} \ldots \sum_{u'_1 \in [R]: u'_1 \neq u_1}  y_{S \cup \{(j,u), (j_1, u'_1), \ldots, (j_t, u'_t)\}}\,.
  \end{align*}
  Hence, we have also that all the inequalities hold if they hold for those with $T=\emptyset$ and $S$ such that $|S| < r$. 
\end{proof}

By the above claim, the constraints of the canonical $r$-rounds 
SA relaxation of the CSP can be simplified to:
\begin{align*}
  \sum_{u\in [R]}  y_{S \cup \{(j,u)\}} & = y_{ S}  
&\forall j\in[n], S\subseteq [n]\times[R]: |S| < r\,,\\[1ex]
y_{S \cup \{(j,u)\}} & \geq 0  &\forall (j,u)\in[n]\times [R], S\subseteq [n]\times[R]: |S| < r\,.
\end{align*}
To see that this is equivalent  to~\eqref{eq:CSP_prog2} observe first that $y_S = 0$ if
$\{(j,u'), (j,u'')\} \subseteq S$. Indeed, by the partition constraint, we have
\begin{align*}
  \sum_{u\in R} y_{\{(j,u'), (j,u'')\} \cup \{(j,u)\}} & = y_{\{(j,u'), (j,u'')\}}\,,
\end{align*}
which implies the constraint $2y_{\{(j,u'), (j,u'')\}}  \leq y_{\{(j,u'),
(j,u'')\}}$. This in turn (together with the non-negativity) implies that
$y_{\{(j,u'), (j,u'')\}} = 0$. Therefore, by again using the partition
constraint, we have $y_S =0$ whenever $\{(j,u'), (j,u'')\} \subseteq S$ and hence we can discard variables of this type.
We now obtain the formulation~\eqref{eq:CSP_prog2} by using variables of type
$X_{(\{j_1, \ldots, j_t\}, (u_1, \ldots, u_t))}$ instead of $y_{\{(j_1, u_1),
(j_2, u_2), \ldots, (j_t, u_t)\}}$. The objective function can be linearized,
provided that the number of rounds is at least the arity of the CSP, that is 
$r \geqslant k$, so that variables for sets of cardinality $k$ are available.

\section{Proof of Claim \ref{claim:completeness}}
\label{app:proofofclaim}
\begin{proof}[Proof of Claim~\ref{claim:completeness}]
Assume that $\Lambda$ satisfies $vw_1,\dots,vw_t$ simultaneously, i.e., \begin{align}
\label{eq:sat}
\pi_{v,w_1}(\Lambda(w_1)) = \dots = \pi_{v,w_t}(\Lambda(w_t)) = \Lambda(v)
\end{align}
and let $C_{x,S}$ and $C_{\bar{x},S}$ be the sub-cubes as in Figure \ref{fig:dist}. According to the new assignment, every variable $\left<w_i,z\right>$ in the support of $C(v,\mc{W}_v,x,S)$ takes the value $z_{\Lambda(w_i)}$. Assume w.l.o.g.~that $\left<w_i,z\right>$ is such that $\pi^{-1}_{v,w_i}(z) \in C_{x,S}$, and let $y\in C_{x,S}$ satisfies $\pi_{v,w_i}(y) = z$. Then we get \begin{align}
\label{eq:subcube}
z_{\Lambda(w_i)} = \pi_{v,w_i}(y)_{\Lambda(w_i)} = y_{\pi_{v,w_i}(\Lambda(w_i))} = y_{\Lambda(v)}
\end{align}
where the last equality follows from (\ref{eq:sat}). We know from the construction of the sub-cube $C_{x,S}$ that for all $j\notin S$ and for all $y \in C_{x,S}$, we have $y_j = x_j$. It then follows that if $\Lambda(v) \notin S$, equation \ref{eq:subcube} yields that \begin{align*}
z_{\Lambda(w_i)} = y_{\Lambda(v)} = x_{\Lambda(v)} && \forall \left<w_i,z \right> \text{ s.t. } \pi^{-1}_{v,w_i}(z) \in C_{x,S}
\end{align*}
Similarly, for the variables  $\left<w_i,z\right>$ with $\pi^{-1}_{v,w}(z) \in C_{\bar{x},S}$, we get that \begin{align*}
z_{\Lambda(w_i)} = y_{\Lambda(v)} = \bar{x}_{\Lambda(v)} && \forall \left<w_i,z \right> \text{ s.t. } \pi^{-1}_{v,w_i}(z) \in C_{\bar{x},S}
\end{align*}
Thus far we proved that if If $\Lambda$ satisfies $vw_1,\dots,vw_t$ simultaneously and $\Lambda(v) \notin S$, then $\psi(\Lambda)$ satisfies $C(v,\mc{W}_v,x,S)$. But we know by construction that $|S| = \epsilon R$, and hence $\Lambda(v) \notin S$ with probability at least $1-\epsilon$.
\end{proof}

\section{Proof of Claim \ref{claim:completeness2}}
\label{app:proofofclaim2}
\begin{proof}[Proof of Claim~\ref{claim:completeness2}]
Assume that $\Lambda$ satisfies $vw_1,\dots,vw_t$ simultaneously, i.e., \begin{align}
\label{eq:sat2}
\pi_{v,w_1}(\Lambda(w_1)) = \dots = \pi_{v,w_t}(\Lambda(w_t)) = \Lambda(v)
\end{align}
and let $C_{x,S_\epsilon}$ be the sub-cube as in Figure \ref{fig:dist2}. For $z \in [q]^R$ with $\pi^{-1}_{v,w_i}(z) \in C_{x,S_\epsilon}$, let $y \in [q]^R$ be such that $\pi_{v,w_i}(y) = z$. Recall that a constraint $C(v,\mc{W}_v,x,S_\epsilon)$ looks as follows: \begin{align}
\label{pred}
\left< w_i,z \oplus \tilde{z}\right> \neq {z}_0 && \forall \,\, 1\leq i \leq t, \forall {z} \text{ such that }\pi^{-1}_{v,w_1}({z}) \in C_{x,S_\epsilon}
\end{align}
We now adopt the functions point of view, i.e., for a $w \in W$, the variables $\left<w,z\right>$ for $z \in [q]^R$ with $z_0$ are the entries of the truth table of a function $f_w$, and according to the new assignment $\Lambda$, $f_w$ is the \emph{folded} dictatorship function of the label of $\Lambda(w)$. 

So if we let $f:=f_{w_i}$ for some $1\leq i \leq t$, and $z:= \left<w_i,z\right>$, we get that \begin{align*}
\left< w_i,{z} \oplus \tilde{z}\right> \neq {z}_0 && \Longleftrightarrow && f(z) \neq 0
\end{align*}
and by our definition of the dictatorship function, the latter is zero iff ${z}_{\Lambda(w_i)} = 0$. But \begin{align}
\label{eq:subcube2}
z_{\Lambda(w_i)} = \pi_{v,w_i}(y)_{\Lambda(w_i)} = y_{\pi_{v,w_i}(\Lambda(w_i))} = y_{\Lambda(v)}
\end{align}
where the last equality follows from (\ref{eq:sat2}). We know from the construction of the sub-cube $C_{x,S_\epsilon}$ that for all $j\notin S_\epsilon$ and for all $y \in C_{x,S_\epsilon}$, we have $y_j = x_j$. It then follows that if $\Lambda(v) \notin S_\epsilon$, equation \ref{eq:subcube2} yields that \begin{align*}
z_{\Lambda(w_i)} = y_{\Lambda(v)} = x_{\Lambda(v)} && \forall \left<w_i,z \right> \text{ s.t. } \pi^{-1}_{v,w_i}(z) \in C_{x,S_\epsilon}
\end{align*}
Moreover, given that $x$ is chosen uniformly at random from $[q]^R$, we get that for any $i\in [R]$, $\Pr_{x\in[q]^R}\left[ x_{i}=0 \right] = \frac{1}{q}$.

Thus far we proved that if If $\Lambda$ satisfies $vw_1,\dots,vw_t$ simultaneously and $\Lambda(v) \notin S$, then $\psi(\Lambda)$ satisfies $C(v,\mc{W}_v,x,S)$ with probability $1-\frac{1}{q}$. But we know by construction that $|S| = \epsilon R$, and hence $\Lambda(v) \notin S$ with probability at least $1-\epsilon$.
\end{proof}

\end{document}